\documentclass[11pt,draftcls, onecolumn]{IEEEtran}
\usepackage[english]{babel}
\usepackage{cite}
\usepackage{amsmath,amssymb,amsfonts}
\usepackage{listings}
\usepackage{graphicx}
\usepackage{textcomp}
\usepackage{xcolor}
\usepackage{booktabs}
\usepackage{multirow}
\usepackage{array}
\usepackage{hyperref}

\usepackage{glossaries-extra}
\setabbreviationstyle[acronym]{long-short}
\glsdisablehyper

\usepackage{cite}
\usepackage{setspace}
\usepackage[T1]{fontenc} % Seleção de códigos de fonte
\usepackage[utf8]{inputenc} % Codificação do documento (conversão automática dos acentos)
\usepackage{pdfpages} % Inclusão de (páginas de) arquivos PDF no documento
\usepackage{comment}

\usepackage{xcolor}
\usepackage{lmodern}

\usepackage{amsmath}
\usepackage{amssymb}
\usepackage{amsthm}
\usepackage{mathtools}
\usepackage{bm}
\usepackage{cases}
\usepackage{optidef}

\usepackage{array}
\usepackage{multicol,multirow,colortbl,arydshln}
\usepackage{float}
\usepackage[caption=false,font=footnotesize]{subfig}
\usepackage{epstopdf}
\usepackage{tikz}
\usepackage{tikz-3dplot}

\usepackage{booktabs}
\usepackage{tabularx}
\usepackage{url,hyperref}
\usepackage{acronym}

\usepackage{algorithm}
\usepackage{algpseudocode}
\usepackage{threeparttable}

\DeclareMathOperator{\diag}{diag}

\newtheorem{lemma}{Lemma}
\newtheorem{theorem}{Theorem}
\newtheorem{proposition}{Proposition}
\newtheorem{corollary}{Corollary}

\definecolor{gold}{rgb}{0.85,.66,0}

\newcommand{\K}{\mathcal{K}} % set of all users
\newcommand{\U}{\mathcal{U}} % scheduled user set
\newcommand{\D}{\mathcal{D}} % serving subarrays set
\newcommand{\eye}{\mathbf{I}}

\newcommand{\E}{\mathbb{E}}
\newcommand{\C}{\mathbb{C}}

\DeclareMathOperator\tr{tr}

\newacronym{MIMO}{MIMO}{multiple-input multiple output}
\newacronym{mMIMO}{mMIMO}{massive MIMO}
\newacronym{XL-MIMO}{XL-MIMO}{extra-large scale MIMO}
\newacronym{PA}{PA}{pilot assignment}
\newacronym{US}{US}{user scheduling}
\newacronym{AWGN}{AWGN}{additive white Gaussian noise}
\newacronym{BS}{BS}{base station}
\newacronym{SA}{SA}{subarray}
\newacronym{AP}{AP}{access point}
\newacronym{ULA}{ULA}{uniform linear array}
\newacronym{UPA}{UPA}{uniform planar array}
\newacronym{LoS}{LoS}{line-of-sight}
\newacronym{NLoS}{NLoS}{non-LoS}
\newacronym{SNR}{SNR}{signal-to-noise ratio}
\newacronym{SINR}{SINR}{signal-to-interference-plus-noise ratio}
\newacronym{LS}{LS}{least-squares}
\newacronym{MSE}{MSE}{mean square error}
\newacronym{MMSE}{MMSE}{minimum mean square error}
\newacronym{P-MMSE}{P-MMSE}{partial MMSE}
\newacronym{L-MMSE}{L-MMSE}{local MMSE}
\newacronym{SE}{SE}{spectral efficiency}
\newacronym{QoS}{QoS}{quality of service}
\newacronym{VR}{VR}{visibility region}
\newacronym{UL}{UL}{uplink}
\newacronym{DL}{DL}{downlink}
\newacronym{mMTC}{mMTC}{massive machine type communications}
\newacronym{CSI}{CSI}{channel state information}
\newacronym{UE}{UE}{user equipment}
\newacronym{CPU}{CPU}{central processing unit}
\newacronym{LSF}{LSF}{large-scale fading}
\newacronym{SSF}{SSF}{small-scale fading}
\newacronym{LSFD}{LSFD}{large-scale fading decoding}
\newacronym{iid}{i.i.d.}{independent and identically distributed}
\newacronym{NAE}{NAE}{normalized absolute error}
\newacronym{MNAE}{MNAE}{mean normalized absolute error}
\newacronym{NMSE}{NMSE}{normalized mean square error}
\newacronym{MR}{MR}{maximum-ratio}
\newacronym{ZF}{ZF}{zero-forcing}
\newacronym{L-ZF}{L-ZF}{local ZF}
\newacronym{PZF}{PZF}{partial ZF}
\newacronym{RZF}{RZF}{regularized ZF}
\newacronym{TDD}{TDD}{time-division-duplexing}
\newacronym{FDD}{FDD}{frequency-division-duplexing}
\newacronym{KKT}{KKT}{Karush–Kuhn–Tucker}
\newacronym{5G}{5G}{fifth-generation}
\newacronym{6G}{6G}{sixth-generation}
\newacronym{RIS}{RIS}{reflective intelligent surface}
\newacronym{UPD}{UPD}{uniform power distance}
\newacronym{UPW}{UPW}{uniform plane wave}
\newacronym{NUPW}{NUPW}{non-uniform plane wave}
\newacronym{USW}{USW}{uniform spherical wave}
\newacronym{UPBW}{UPBW}{uniform parabolic wave}
\newacronym{NUSW}{NUSW}{non-uniform spherical wave}
\newacronym{NUEW}{NUEW}{non-uniform evanescent wave}
\newacronym{HD}{HD}{half‑duplex}
\newacronym{LTE}{LTE}{long term evolution}
\newacronym{RF}{RF}{radio frequency}
\newacronym{eMBB}{eMBB}{enhanced mobile broadband}
\newacronym{URLLC}{URLLC}{ultra-reliable and low latency communications}
\newacronym{ISAC}{ISAC}{integrated sensing and communication}
\newacronym{AI}{AI}{artificial intelligence}
\newacronym{ELAA}{ELAA}{extremely large aperture array}
\newacronym{CoMP}{CoMP}{coordinated multipoint}
\newacronym{C-RAN}{C-RAN}{cloud radio access network}
\newacronym{PDF}{PDF}{probability density function}
\newacronym{AoA}{AoA}{angle of arrival}
\newacronym{DoA}{DoA}{direction of arrival}
\newacronym{AoD}{AoD}{angle of departure}
\newacronym{DoD}{DoD}{direction of departure}
\newacronym{GA}{GA}{genetic algorithm}
\begin{document}

\title{\bf %Deterministic SINR Approximations and 
Joint Subarray Selection, User Scheduling, and Pilot Assignment for XL-MIMO}% Systems}

\author{
\IEEEauthorblockN{%\gau
{Gabriel Avanzi Ubiali}\IEEEauthorrefmark{1}, %\jcm
{José Carlos Marinello Filho}\IEEEauthorrefmark{2}, %\ta
{Taufik Abrão}\IEEEauthorrefmark{1}}\\
\IEEEauthorblockA{\IEEEauthorrefmark{1} \small Department of Electrical Engineering, State University of Londrina (UEL), Londrina, Brazil}\\
\IEEEauthorblockA{\IEEEauthorrefmark{2}\small Department of Electrical Engineering, Federal Technological University of Paraná (UTFPR), Cornélio Procópio, Brazil}\\
E-mail: gabriel.ubiali@uel.br, jcmarinello@utfpr.edu.br, taufik@uel.br
}

\maketitle

\begin{abstract}
\Gls{XL-MIMO} is a key technology for meeting \gls{6G} requirements for high-rate connectivity and uniform \gls{QoS}; however, its deployment is challenged by the prohibitive complexity of resource management based on instantaneous \gls{CSI}. To address this intractability, this work derives novel closed-form deterministic \gls{SINR} expressions for both centralized and distributed uplink operations. Valid for Rician fading channels with \gls{MMSE} receive combining and \gls{MMSE} channel estimation, these expressions depend exclusively on long-term channel statistics, providing a tractable alternative to computationally expensive instantaneous \gls{CSI}-driven optimization. Building on these results, we develop statistical-\gls{CSI}-based algorithms for joint subarray selection, users scheduling, and pilot assignment, leveraging the derived \gls{SINR} approximations to maximize the minimum \gls{SE} among scheduled users while preserving computational tractability. The proposed framework exploits the spatial sparsity of \gls{UE} \glspl{VR} to enable more aggressive pilot reuse than is possible in conventional massive MIMO. Numerical results validate the high accuracy of the derived \gls{SINR} approximations and demonstrate that the proposed algorithms significantly enhance fairness and throughput in crowded scenarios.
\end{abstract}


%=================== \input{introduction} =============================
\section{Introduction}
\label{sec:introduction}

The evolution toward \gls{6G} networks demands unprecedented \gls{SE} and link reliability \cite{2024_XLMIMO_tutorial,6824752,7894280,7414384}. \Gls{XL-MIMO}, characterized by antenna deployments with hundreds or thousands of elements, have emerged as a key enabling technology to meet these requirements \cite{CFbook,Bjornson2019_WhatIsNext,ChannelEstimation_XLMIMO,9529197}. Unlike co-located massive MIMO, \gls{XL-MIMO} features antenna arrays distributed over large areas (typically organized into discrete subarrays), introducing near-field effects and spatial non-stationarity that fundamentally alter system design principles \cite{2024_XLMIMO_tutorial,9903389,carvalho2020nonstationarities}.

%While \gls{XL-MIMO} provides unprecedented spatial degrees of freedom, realizing its potential requires addressing three coupled challenges: (i) scalability under high user density, (ii) interference management under aggressive pilot reuse, and (iii) reduction of channel estimation overhead.

%As the number of antennas at the \gls{BS} increases in \gls{XL-MIMO} systems, the associated baseband processing can become computationally prohibitive. A practical and widely adopted solution is to partition the array into multiple subarrays, each comprising a small to moderate number of antennas and typically serving only a subset of \glspl{UE} located within its \gls{VR}.

In \gls{XL-MIMO} deployments, each \gls{UE} is typically close to only a few subarrays. Consequently, serving each \gls{UE} with all antennas is neither necessary nor scalable. Intelligent subarray selection ensures \glspl{UE} are served only by subarrays contributing significantly to their signal, preserving service quality while reducing computational complexity \cite{Marinello2020,Location-Based_VR_Recognition,CFbook}. Furthermore, in dense scenarios where \glspl{UE} outnumber spatial degrees of freedom, user scheduling is essential to manage interference and satisfy \gls{QoS} constraints \cite{JOAO2023,10857377}.% Consequently, selecting an appropriate subset of \glspl{UE} that can be served concurrently while satisfying \gls{QoS} constraints is crucial for improving fairness, throughput, and computational efficiency\cite{JOAO2023}.

A fully centralized uplink architecture, where a \gls{CPU} performs joint detection, achieves the highest \gls{SE} by enabling global interference suppression \cite{Bjornson2020_making_centralizedCF_competitive}. However, its deployment is often limited by fronthaul constraints. A more scalable alternative is distributed processing, where channel estimation and data detection are performed locally at the subarrays. This enhances implementation flexibility and scalability, as new subarrays can be integrated without upgrading the \gls{CPU} \cite[Sec. 5.2]{CFbook}.

Coherent processing relies on accurate \gls{CSI}. While ideally \glspl{UE} would use orthogonal pilot sequences, pilot reuse is often needed due to limited channel coherence blocks \cite{CFbook,CF_PA_HungarianAlgorithm}. Pilot sharing reduces estimation accuracy and induces statistical dependence between channel estimates. Fortunately, the spatial sparsity of \glspl{VR} in \gls{XL-MIMO} permits more aggressive pilot reuse than in conventional massive MIMO, %provided that resource allocation is designed to mitigate the resulting pilot contamination \cite{CFbook}.
where every scheduled \gls{UE} is %effectively
``seen'' by all antennas.

We develop statistical-\gls{CSI}-based algorithms for joint subarray selection, user scheduling, and pilot assignment. Our key innovation lies in the use of deterministic \gls{SINR} approximations derived from channel statistics. Because these expressions depend on slowly varying statistical \gls{CSI} rather than instantaneous \gls{CSI}, they remain valid over multiple coherence blocks. This significantly reduces computational and signaling overhead, providing a highly efficient framework for resource allocation. Furthermore, channel hardening in \gls{XL-MIMO} ensures that statistical-\gls{CSI}-driven algorithm achieves \gls{SE} comparable to schemes relying on instantaneous \gls{CSI}.

%Deterministic \gls{SINR} approximations are powerful tools for system analysis and design because they depend only on channel statistics, which vary far more slowly than the instantaneous channel. This approach significantly reduces computational load and eliminates the need for frequent instantaneous \gls{CSI} updates from the subarrays to the \gls{CPU}, allowing resource allocation algorithms to operate over extended periods with a single \gls{SINR} calculation.

%We derive novel deterministic \gls{SINR} expressions valid for \gls{XL-MIMO} channels incorporating \gls{LoS} and spatially correlated \gls{NLoS} components. Furthermore, this work advances beyond the common assumptions of perfect \gls{CSI} and simpler \gls{MR} or \gls{ZF} combiners. Unlike existing studies, our framework uniquely accounts for channel estimation errors and the use of \gls{MMSE} receive combining, providing a robust foundation for practical \gls{XL-MIMO} resource management.

\subsection{Related Works}
\label{subsec:intro:related_works}

Deterministic \gls{SINR} expressions have been extensively investigated in multi-user massive MIMO, \gls{XL-MIMO}, and cell-free architectures. However, most existing expressions are derived under spatially uncorrelated Rayleigh fading with simple \gls{MR} processing, which performs poorly in interference-limited regimes. In realistic \gls{6G} micro-urban environments, channels are characterized by a combination of \gls{LoS} and \gls{NLoS} components, best captured by the Rician fading model. Furthermore, while \gls{MR} is analytically convenient, \gls{MMSE} combining is essential for interference suppression in dense networks \cite{Bjornson2020_making_centralizedCF_competitive}.

The literature regarding deterministic \gls{SINR} can be categorized by channel models and combining schemes. Expressions for spatially uncorrelated Rician fading are available in \cite{9618945,8952782,10817111,10858168,Hosany2022,10192084,9099874}, but assume \gls{MR} processing or perfect \gls{CSI}. Conversely, works focusing on spatially correlated Rayleigh fading, such as \cite{10770206,statistical_CSI_based_algorithms,XL_MIMO_letter} for perfect \gls{CSI} and \cite{9174860,9042221,9462554,9018056,10201892,Li2024,ergodicSINRproof} for imperfect \gls{CSI}, neglect the \gls{LoS} component and generally restrict analysis to \gls{MR} combiners, thereby limiting their applicability in dense deployments. Regarding other combining schemes, \cite{statistical_CSI_based_algorithms} and \cite{ergodicSINRproof} investigate \gls{ZF} and \gls{RZF} variants, under perfect and imperfect \gls{CSI}, respectively.% However, they still omit the \gls{LoS} component.

%Deterministic \gls{SINR} expressions for spatially correlated Rayleigh fading with \glspl{VR} are presented in \cite{XL_MIMO_letter} for \gls{MR} and \gls{ZF} precoding, assuming perfect channel estimation.

More recent efforts have addressed spatially correlated Rician fading with imperfect \gls{CSI} and \gls{MR} processing\cite{9099874,10507003,9684861,9276421,10041946,Kaur2023}. A notable exception is \cite{Kaur2023}, which analyzes \gls{RZF} precoding. However, the resulting expressions are not fully closed-form, requiring iterative solutions of coupled fixed-point equations% via numerical iteration, although the authors demonstrate the existence and uniqueness of the solution
.

To the best of our knowledge, no existing work provides closed-form deterministic \gls{SINR} expressions that jointly account for: (i) spatially correlated Rician fading, (ii) under imperfect \gls{CSI}, where channel estimates are obtained via \gls{MMSE} estimation, and (iii) \gls{MMSE} receive combining. This paper fills this gap, providing a robust analytical framework for both centralized and distributed operations.

%---------------------------------
\subsection{Contributions}
\label{subsec:intro:contributions}
%---------------------------------

This paper's contributions are summarized as:

\begin{itemize}
    \item \textbf{Novel closed-form deterministic SINR expressions:} We derive original, closed-form deterministic expressions for the uplink \gls{SINR} under both centralized and distributed \gls{XL-MIMO} operations. These expressions are unique in that they jointly account for spatially correlated Rician fading, \gls{MMSE} receive combining, and imperfect \gls{CSI} obtained through \gls{MMSE} channel estimation. By depending exclusively on statistical \gls{CSI},  they enable rapid performance evaluation and provide a tractable analytical foundation for optimization. Numerical results validate that these approximations closely track the true \gls{SINR}.
    \item \textbf{Statistical-CSI-driven resource allocation:} We develop algorithms for joint subarray selection, user scheduling, and pilot assignment that leverage the derived \gls{SINR} approximations to maximize the minimum \gls{SE} among the scheduled \glspl{UE}. By utilizing only slowly-varying channel statistics, the proposed methods drastically reduce computational complexity and signaling overhead compared to instantaneous-\gls{CSI} benchmarks while maintaining near-optimal performance. Numerical results demonstrate that our statistical approach closely matches the performance of instantaneous-\gls{CSI}-based strategies.
    %\item \textbf{High-fidelity \gls{XL-MIMO} channel framework:} We develop a robust channel model that captures the essential physics of \gls{XL-MIMO}, including spatial non-stationarities (visibility regions) and a hybrid propagation environment consisting of \gls{LoS} and correlated \gls{NLoS} components.
\end{itemize}

\section{XL-MIMO Channel and System Models}
\label{sec:system_model}
%-----------------------------
%-----------------------------

%Consider a \gls{BS} architecture where a \gls{CPU} coordinates $L$ subarrays to serve $K$ single-antenna \glspl{UE}. Each subarray is equipped with a \gls{UPA} consisting of $M$ antenna elements. Let $\mathcal{L} = \{1, \dots, L\}$ and $\mathcal{K} = \{1, \dots, K\}$ denote the index sets of the subarrays and \glspl{UE}, respectively. In a given coherence block, a subset $\mathcal{U} \subseteq \mathcal{K}$ of $U = |\mathcal{U}|$ \glspl{UE} is scheduled for transmission. Each scheduled \gls{UE} $k \in \mathcal{U}$ is served by a specific subset of subarrays $\mathcal{D}_k \subseteq \mathcal{L}$, with cardinality $L_k = |\mathcal{D}_k|$. Conversely, we define $\mathcal{U}_l = \{k \in \mathcal{U} : l \in \mathcal{D}_k\}$ as the set of $U_l = |\mathcal{U}_l|$ \glspl{UE} served by subarray $l \in \mathcal{L}$.

Consider a system with $L$ subarrays, coordinated by a \gls{CPU}, serving $K$ single-antenna \glspl{UE}. Each subarray utilizes a \gls{UPA} with $M$ antennas. Defining the sets of all subarrays and \glspl{UE} as $\D = \{1, \dots, L\}$ and $\K = \{1, \dots, K\}$, respectively, we denote $\mathcal{U} \subseteq \mathcal{K}$ as the set of scheduled \glspl{UE}.

Each \gls{UE} $k \in \mathcal{U}$ is served by a subarray subset $\D_k \subseteq \D$ of cardinality $L_k = \lvert \D_k \rvert$, while subarray $l$ serves a \gls{UE} subset $\U_l \subseteq \U$ of size $U_l = \lvert \U_l \rvert$. For convenience, define the selection matrix $\mathbf{D}_k = \operatorname{blkdiag} \left(\mathbf{D}_{k1}, %\mathbf{D}_{k2}, 
\ldots, \mathbf{D}_{kL}\right)$, where
\begin{equation}
    \mathbf{D}_{kl} = 
    \begin{cases}
        \mathbf{I}_M & \text{if } l \in \D_k, \\
        \mathbf{0}_M & \text{otherwise.}
    \end{cases}
\label{eq:D_kl}
\end{equation}

We define three subarray selection criteria: \textit{random}; \textit{LSF-based}, which prioritizes signal strength by selecting subarrays with the highest \gls{LSF} channel gains; and \textit{SINR-based}, which optimizes the trade-off between signal power and interference by selecting subarrays with the highest local \glspl{SINR}.

During uplink data transmission, subarray $l$ receives the signal
\begin{equation}
    \mathbf{y}_l^\mathrm{ul} = \sum_{i\in\U} \mathbf{h}_{il} x_i + \mathbf{n}_l^\mathrm{ul},
\label{eq:y_l_UL}
\end{equation}
where $\mathbf{h}_{il}$ is the channel vector between \gls{UE} $i$ and subarray $l$, $x_i\in\C$ is the zero-mean symbol transmitted by \gls{UE} $i$, with transmit power $p_i = \E\{ \lvert x_i \rvert^2 \}$, and $\mathbf{n}_l^\mathrm{ul} \sim \mathcal{N}_\C(\mathbf{0}, \sigma_\mathrm{n}^2 \mathbf{I}_M)$ denotes the \gls{AWGN} vector.

The channel vector between \gls{UE} $k$ and subarray $l$ comprises a deterministic \gls{LoS} component, $\overline{\mathbf{h}}_{kl}$, and a stochastic \gls{NLoS} component, $\check{\mathbf{h}}_{kl} \sim \mathcal{N}_\mathbb{C} (\mathbf{0}, \mathbf{R}_{kl})$, where $\mathbf{R}_{kl} \in \mathbb{C}^{M \times M}$ is the spatial covariance matrix. It follows the distribution $\mathbf{h}_{kl} \sim \mathcal{N}_\mathbb{C} (\omega_{kl} \overline{\mathbf{h}}_{kl}, \mathbf{R}_{kl})$, where $\omega_{kl} \in \{0, 1\}$ is a Bernoulli random variable indicating the presence of a \gls{LoS} path. The average channel gain between \gls{UE} $k$ and the antennas of subarray $l$ is
\begin{equation}
    \beta_{kl} = \frac{1}{M} \E\{\lVert \mathbf{h}_{kl} \rVert^2\} = \frac{1}{M} \tr(\mathbf{Q}_{kl}),
\end{equation}
where $\mathbf{Q}_{kl} = \E\{\mathbf{h}_{kl} \mathbf{h}_{kl}^\mathrm{H}\} = \overline{\mathbf{h}}_{kl} \overline{\mathbf{h}}_{kl}^\mathrm{H} + \mathbf{R}_{kl}$ is the channel correlation matrix. Similarly, the channel gain associated only to the \gls{NLoS} component is $\check{\beta}_{kl} = \frac{1}{M} \tr(\mathbf{R}_{kl})$. The average \gls{LSF} channel gain across all $L$ subarrays for \gls{UE} $k$ is $\beta_k = \frac{1}{L} \sum_{l=1}^L \beta_{kl}$.

In modular \gls{XL-MIMO}, the large inter–subarray spacing justifies modeling the \gls{NLoS} channels of different subarrays as statistically independent. Hence, $\E\{\check{\mathbf{h}}_{kl} \check{\mathbf{h}}_{kl'}^\mathrm{H}\} = \mathbf{R}_{kl} \delta_{ll'}$, where $\delta_{ll'}$ is the Kronecker delta. As a result, the overall channel covariance matrix for user $k$ %, defined as $\mathbf{R}_k = \E\{ \check{\mathbf{h}}_k \check{\mathbf{h}}_k^\mathrm{H} \}$, where $\check{\mathbf{h}}_k = [\check{\mathbf{h}}_{k1}^\mathrm{T}, \ldots, \check{\mathbf{h}}_{kL}^\mathrm{T}]^\mathrm{T}$, 
is taken as a block‐diagonal matrix, $\mathbf{R}_k = \operatorname{blkdiag} \left( \mathbf{R}_{k1}, \mathbf{R}_{k2}, \ldots, \mathbf{R}_{kL} \right)$, and the collective channel vector $\mathbf{h}_k = [\mathbf{h}_{k1}^\mathrm{T}, \ldots, \mathbf{h}_{kL}^\mathrm{T}]^\mathrm{T} \in \C^{ML}$ is distributed as $\mathbf{h}_k \sim \mathcal{N}_\mathbb{C} (\overline{\mathbf{h}}_k, \mathbf{R}_k)$, where $\overline{\mathbf{h}}_k = [\overline{\mathbf{h}}_{k1}^\mathrm{T}, \ldots, \overline{\mathbf{h}}_{kL}^\mathrm{T}]^\mathrm{T}$.

% Collecting the vectors of the \(K\) \glspl{UE} results in the subarray matrices
% %
% \begin{align*}
%     \mathbf{H}_l &=
%     \left[ \mathbf{h}_{1l}, \dots, \mathbf{h}_{Kl} \right] \in\mathbb{C}^{M \times K},
%     \\
%     \widetilde{\mathbf{H}}_l &=
%     \left[ \widetilde{\mathbf{h}}_{1l}, \dots, \widetilde{\mathbf{h}}_{Kl} \right] =
%     \mathbf{H}_l - \widehat{\mathbf{H}}_l
%     \in \mathbb{C}^{M \times K},
% \end{align*}

% The global combining matrices and channel matrices are given by:
% %
% \begin{align*}
%     \mathbf{H} &=
%     \begin{bmatrix}
%         \mathbf{H}_{1} \\[0.3em]
%         \vdots \\[0.3em]
%         \mathbf{H}_{L}
%     \end{bmatrix}
%     = \left[ \mathbf{h}_{1}, \dots, \mathbf{h}_{K} \right]
%     \in \C^{ML \times K},\\
%     \widetilde{\mathbf{H}} &=
%     \begin{bmatrix}
%         \widetilde{\mathbf{H}}_{1} \\[0.3em]
%         \vdots \\[0.3em]
%         \widetilde{\mathbf{H}}_{L}
%     \end{bmatrix}
%     = \left[ \widetilde{\mathbf{h}}_{1}, \dots, \widetilde{\mathbf{h}}_{K} \right]
%     = \mathbf{H} - \widehat{\mathbf{H}}
%     \in\C^{ML\times K}.
% \end{align*}

%----------------------------------
\subsection{MMSE Channel Estimation}
\label{subsec:SINR:system_model:channel_estimation:MMSE}
%----------------------------------

The \gls{CPU} is assumed to have perfect knowledge of the channel statistics (mean vectors and covariance matrices), as these remain approximately constant over numerous coherence blocks.\footnote{Such statistics can be acquired during connection setup %These statistics can be obtained via pilot-based training during connection setup, by computing sample means and sample covariance matrices over multiple coherence blocks
\cite{LOS_NLOS_2022}.} Under \gls{TDD} operation, instantaneous channels are estimated via uplink pilot signaling. Each coherence block of length $\tau_\mathrm{c}$ consists of $\tau_\mathrm{p}$ symbols reserved for pilots and $\tau_\mathrm{c} - \tau_\mathrm{p}$ symbols dedicated to data transmission. This study focuses exclusively on the uplink data transmission. The pilot assigned to \gls{UE} \(k\) is denoted by \(t_k\), and its corresponding sequence is \(\boldsymbol{\phi}_{t_k} \in \C^{\tau_\mathrm{p}}\). The sequences %have unit-magnitude elements \cite{massivemimobook} and 
form an orthogonal set:
\begin{equation}
    \boldsymbol{\phi}_{t_k}^\mathrm{H} \boldsymbol{\phi}_{t_i} = 
    \begin{cases}
        \tau_\mathrm{p}, & t_i = t_k, \\
        0, & \text{otherwise}.
    \end{cases}
\label{eq:pilot_orthogonality}
\end{equation}

Let \(\mathcal{P}_t = \{k\in\K \mid t_k = t\}\) denote the set of \glspl{UE} using pilot \(t\). The \gls{MMSE} estimator exploits the channel statistics to separate pilot-sharing \glspl{UE} and minimize the estimation \gls{MSE}. Denote by \(\widehat{\mathbf{h}}_{kl}\) the \gls{MMSE} estimate of \(\mathbf{h}_{kl}\) and by \(\widetilde{\mathbf{h}}_{kl} = \mathbf{h}_{kl} - \widehat{\mathbf{h}}_{kl}\) the corresponding estimation error. They are independent and distributed as \cite{[2019]Massive_MIMO_With_Spatially_Correlated_Rician_Fading_Channels}:
\begin{align}
    \widehat{\mathbf{h}}_{kl}
    &\sim \mathcal{N}_\C\left( \overline{\mathbf{h}}_{kl}, \mathbf{R}_{kl} - \mathbf{C}_{kl} \right),
\label{eq:hkl_est_distribution}\\
    \widetilde{\mathbf{h}}_{kl}
    &\sim \mathcal{N}_\C\left( \mathbf{0}_{M}, \mathbf{C}_{kl} \right),
\label{eq:hkl_err_distribution}
\end{align}
where the error covariance matrix is
\begin{equation}
    \mathbf{C}_{kl} = \mathbf{R}_{kl} - p_k \tau_\mathrm{p} \mathbf{R}_{kl} \boldsymbol{\Psi}_{t_kl}^{-1} \mathbf{R}_{kl},
\label{eq:Ckl_MMSE}
\end{equation}
and $\boldsymbol{\Psi}_{tl} = \sum_{i\in\mathcal{P}_t} p_i \tau_\mathrm{p} \mathbf{R}_{il} + \sigma_\mathrm{n}^2 \eye_M$.

The estimation accuracy can be measured by the \gls{NMSE}, evaluated in subarray $l$ and in the entire array by, respectively \cite{CFbook}:
\begin{align}
    \gamma_{kl}
    &= \frac{ \E \{ \lVert \widetilde{\mathbf{h}}_{kl} \rVert^2 \} }{ \E \{ \lVert \mathbf{h}_{kl} \rVert^2 \} }
    = \frac{\tr( \mathbf{C}_{kl} )}{\tr( \mathbf{Q}_{kl} )},
\label{eq:hkl_est:NMSE} \\
    \gamma_k
    &= \frac{\E\{\lVert \widetilde{\mathbf{h}}_k \rVert^2\}}{\E\{\lVert \mathbf{h}_k \rVert^2\}}
    = \frac{\tr(\mathbf{C}_k)}{\tr(\mathbf{Q}_k)},
\label{eq:hk_est:NMSE}
\end{align}
where $\mathbf{C}_k = \operatorname{blkdiag} \left( \mathbf{C}_{k1}, \ldots, \mathbf{C}_{kL} \right)$, $\widetilde{\mathbf{h}}_k = [\widetilde{\mathbf{h}}_{k1}^\mathrm{T}, \ldots, \widetilde{\mathbf{h}}_{kL}^\mathrm{T}]^\mathrm{T}$, and $\mathbf{Q}_k = \E\{\mathbf{h}_k \mathbf{h}_k^\mathrm{H}\} = \overline{\mathbf{h}}_k \overline{\mathbf{h}}_k^\mathrm{H} + \mathbf{R}_k$.

%----------------------------
\subsection{Centralized Uplink Operation}
\label{subsec:system_model:centralized}

In centralized operation, the \gls{CPU} computes the collective receive combiner $\mathbf{v}_k = [\mathbf{v}_{k1}^\mathrm{T}, \ldots, \mathbf{v}_{kL}^\mathrm{T} ]^\mathrm{T}
$ and jointly processes the uplink signals from the selected subarrays to estimate $x_k$, resulting in $\widehat{x}_k %= \sum_{l\in\D_k} \mathbf{v}_{kl}^\mathrm{H} \mathbf{y}_l^\mathrm{ul}
= \mathbf{v}_k^\mathrm{H} \mathbf{D}_k \mathbf{y}^\mathrm{ul}$%\footnote{Only the subarrays in $\D_k$ participate in estimating $x_k$, while the others are ignored via the selection matrix $\mathbf{D}_k$.}
, where $\mathbf{y}^\mathrm{ul} = [(\mathbf{y}_1^\mathrm{ul})^\mathrm{T} ,\ldots, (\mathbf{y}_L^\mathrm{ul})^\mathrm{T}]^\mathrm{T}$. The instantaneous, centralized uplink \gls{SINR} for \gls{UE} $k$ is given by\cite[Ch.~5]{CFbook}
\begin{align}
    \Gamma_k^\mathrm{cent}
    %&= \frac{p_k \lvert \mathbf{v}_k^\mathrm{H} \mathbf{D}_k \widehat{\mathbf{h}}_k \rvert^2}{ \displaystyle \sum_{i \in \U \setminus \{k\}} p_i \lvert \mathbf{v}_k^\mathrm{H} \mathbf{D}_k \widehat{\mathbf{h}}_i \rvert^2 + \sum_{i \in \U} p_i \mathbf{v}_k^\mathrm{H} \mathbf{D}_k \mathbf{C}_i \mathbf{D}_k \mathbf{v}_k + \sigma_\mathrm{n}^2 \left\lVert \mathbf{D}_k \mathbf{v}_k \right\rVert^2 } \label{eq:cent:SINRk_a} \\
    &= \frac{p_k \lvert \mathbf{v}_k^\mathrm{H} \mathbf{D}_k \widehat{\mathbf{h}}_k \rvert^2}{ \displaystyle \mathbf{v}_k^\mathrm{H} \mathbf{D}_k \mathbf{Z}_k \mathbf{D}_k \mathbf{v}_k },
\label{eq:cent:SINRk}
\end{align}
where $\widehat{\mathbf{h}}_k = [\widehat{\mathbf{h}}_{k1}^\mathrm{T}, \dots, \widehat{\mathbf{h}}_{kL}^\mathrm{T}]^\mathrm{T}$,
\begin{equation}
    \mathbf{Z}_k = \mathbf{W} - p_k \mathbf{D}_k \widehat{\mathbf{h}}_k \widehat{\mathbf{h}}_k^\mathrm{H} \mathbf{D}_k,
\end{equation}
and
\begin{equation}
    \mathbf{W} = \sum_{i \in \U} p_i \mathbf{D}_k (\widehat{\mathbf{h}}_i \widehat{\mathbf{h}}_i^\mathrm{H} + \mathbf{C}_i) \mathbf{D}_k + \sigma_\mathrm{n}^2 \eye_{ML}.
\end{equation}

This \gls{SINR} is maximized by the \gls{MMSE} combining vector, which is given by \cite[Ch.~5]{CFbook}
\begin{equation}
    \mathbf{v}_k
    = p_k \mathbf{W}^{-1} \mathbf{D}_k \widehat{\mathbf{h}}_k,
\label{eq:cent:vk_MMSE}
\end{equation}
and yields the maximum \gls{SINR} value
\begin{align}
    \Gamma_k^\mathrm{cent}
    &= p_k \widehat{\mathbf{h}}_k^\mathrm{H} \mathbf{D}_k \mathbf{Z}_k^{-1} \mathbf{D}_k \widehat{\mathbf{h}}_k
    %\\
    %&= p_k \tr\left( \mathbf{D}_k \mathbf{Z}_k^{-1} \mathbf{D}_k \widehat{\mathbf{h}}_k \widehat{\mathbf{h}}_k^\mathrm{H} \mathbf{D}_k \right)
    .
\label{eq:cent:SINRk_MMSE}
\end{align}

%\gls{MMSE} combiner is scalable as long as the number of scheduled \glspl{UE}, $U$, stays finite as $K \to \infty$. %\gls{CPU} needs to compute all the $U$ channel estimates $\{ \widehat{\mathbf{h}}_{il} \,:\, i \in \U \}$ corresponding to any subarray $l$ that is serving \gls{UE} $k$, \textit{i.e.}, subarrays with index $l \in \D_k$.

The resulting \gls{SE} is given by $\eta_k^\mathrm{cent} = g(\Gamma_k^\mathrm{cent})$, where
\begin{equation}
    g(x) = \left(1 - \frac{\tau_\mathrm{p}}{\tau_\mathrm{c}}\right) \log_2 (1 + x).
\label{eq:g}
\end{equation}

As we focus exclusively on the uplink, the \gls{SE} expression assumes that each channel coherence block is divided into uplink pilot transmission and uplink data transmission, ignoring downlink slots.% Alternatively, one may assume that the downlink SINR equals the uplink SINR, and reuse the same expressions.

%-------------------------------------
\subsection{Distributed Uplink Operation}
\label{subsec:system_model:distributed}
%-------------------------------------

In the distributed scheme, each serving subarray $l \in \D_k$ first estimates the uplink symbol of \gls{UE} $k$ using the local combining vector $\mathbf{v}_{kl}$ as:
\begin{equation}
    \hat{x}_{kl} = \mathbf{v}_{kl}^\mathrm{H} \mathbf{y}_l^\mathrm{ul}, \qquad l \in \D_k.
\label{eq:dist:xkl_est}
\end{equation}

Subsequently, the \gls{CPU} forms the final estimate $\hat{x}_k$ by aggregating these local estimates:
\begin{equation}
    \hat{x}_k = \sum_{l\in\D_k} \mu_{kl} \, \hat{x}_{kl}.
\label{eq:dist:xk_est}
\end{equation}

The weights satisfy $\sum_{l \in \mathcal{D}_k} \mu_{kl} = 1$ for all $k \in \mathcal{U}$, where $\mu_{kl} > 0$ if $l \in \mathcal{D}_k$, and $\mu_{kl} = 0$ otherwise.

%Figure~\ref{fig:distributed} illustrates this operation for two scheduled \glspl{UE} with serving sets \(\D_1 = \{1,2\}\) and \(\D_2 = \{2,3\}\). In contrast to the centralized case in Figure~\ref{fig:centralized}, where symbol detection is carried out in a single step at the \gls{CPU}, the distributed architecture first produces local estimates \(\hat{x}_{11}\) and \(\hat{x}_{12}\) of \(x_1\) at subarrays 1 and 2, and \(\hat{x}_{22}\) and \(\hat{x}_{23}\) of \(x_2\) at subarrays 2 and 3, which are then forwarded to the \gls{CPU} and combined into the final estimates \(\hat{x}_1\) and \(\hat{x}_2\).

The instantaneous local \gls{SINR} for \gls{UE} $k$ evaluated at a subarray $l$ that serves this \gls{UE} is given by\cite[Ch.~5]{CFbook}
\begin{align}
    \Gamma_{kl}^\mathrm{dist}
    &= \frac{p_k \left\lvert \mathbf{v}_{kl}^\mathrm{H} \widehat{\mathbf{h}}_{kl} \right\rvert^2
    }{\displaystyle
    \sum_{\substack{i\in\U\\i\neq k}}
    p_i \left\lvert \mathbf{v}_{kl}^\mathrm{H} \widehat{\mathbf{h}}_{il} \right\rvert^2 + \sum_{i\in\U} p_i \mathbf{v}_{kl}^\mathrm{H} \mathbf{C}_{il} \mathbf{v}_{kl} + \sigma_\mathrm{n}^2 \left\lVert \mathbf{v}_{kl} \right\rVert^2 }.
\label{eq:dist:SINRkl}
% \\[0.75em]
%     &= \frac{ p_k \left\lvert \mathbf{v}_{kl}^\mathrm{H} \widehat{\mathbf{h}}_{kl} \right\rvert^2
%     }{ \displaystyle
%     \mathbf{v}_{kl}^\mathrm{H}
%     \left( \sum_{\substack{i\in\U\\i\neq k}}
%     p_i \widehat{\mathbf{h}}_{il} \widehat{\mathbf{h}}_{il}^\mathrm{H} + \sum_{i\in\U} p_i \mathbf{C}_{il} + \sigma_\mathrm{n}^2 \eye_M \right) \mathbf{v}_{kl}
%     }.
% \label{eq:dist:SINRkl_b}
\end{align}

This \gls{SINR} is maximized by the \gls{L-MMSE} combiner, given by\cite[page 149]{CFbook}
\begin{equation}
    \mathbf{v}_{kl} =
    p_k \mathbf{W}_l^{-1} %\mathbf{D}_{kl} 
    \widehat{\mathbf{h}}_{kl},
\label{eq:dist:vkl_LMMSE}
\end{equation}
which yields the maximum value
\begin{align}
    \Gamma_{kl}^\mathrm{dist}
    = p_k \widehat{\mathbf{h}}_{kl}^\mathrm{H} \mathbf{Z}_{kl}^{-1} \widehat{\mathbf{h}}_{kl}
    = p_k \tr \left( \mathbf{Z}_{kl}^{-1} \widehat{\mathbf{h}}_{kl} \widehat{\mathbf{h}}_{kl}^\mathrm{H} \right),
\label{eq:dist:SINRkl_LMMSE}
\end{align}
where
\begin{equation}
    \mathbf{W}_l
    = \sum_{i\in\U} p_i (\widehat{\mathbf{h}}_{il} \widehat{\mathbf{h}}_{il}^\mathrm{H} + \mathbf{C}_{il}) + \sigma_\mathrm{n}^2 \eye_M
\label{eq:dist:Wl}
\end{equation}
and
\begin{equation}
    \mathbf{Z}_{kl}
    = \mathbf{W}_l - p_k \widehat{\mathbf{h}}_{kl} \widehat{\mathbf{h}}_{kl}^\mathrm{H}
    %= \sum_{\substack{i\in\U\\i\neq k}} p_i \widehat{\mathbf{h}}_{il} \widehat{\mathbf{h}}_{il}^\mathrm{H} + \sum_{i\in\U} p_i \mathbf{C}_{il} + \sigma_\mathrm{n}^2 \eye_M
    .
\label{eq:dist:Zkl}
\end{equation}

%Although subarray $l$ only estimates the data signals of the $U_l$ \glspl{UE} it serves, computing the \gls{L-MMSE} combiner for each of those \glspl{UE} requires this subarray to estimate the channels of all $U$ scheduled \glspl{UE}. Thus, the \gls{L-MMSE} combiner is scalable only if $U$ stays finite as $K \to \infty$.

The global \gls{SINR} for \gls{UE} $k$ under distributed operation, denoted as $\Gamma_k^\mathrm{dist}$, is given by \eqref{eq:dist:SINRk}, and the resulting \gls{SE} is given by $\eta_k^\mathrm{dist} = g(\Gamma_k^\mathrm{dist})$.

\begin{figure*}[!t]
\begin{equation}
    \Gamma_k^\mathrm{dist}
    = \frac{
        p_k \left\lvert
        \displaystyle \sum_{l\in\D_k} \mu_{kl} \mathbf{v}_{kl}^\mathrm{H} \widehat{\mathbf{h}}_{kl}
        \right\rvert^2}{
        \displaystyle
        \sum_{i\in\U \setminus\{k\}} p_i
            \left\lvert \sum_{l\in\D_k} \mu_{kl} \mathbf{v}_{kl}^\mathrm{H} \widehat{\mathbf{h}}_{il}\right \rvert^2
        + \sum_{i\in\U} p_i \sum_{l\in\D_k} \mu_{kl}^2
            \mathbf{v}_{kl}^\mathrm{H} \mathbf{C}_{il} \mathbf{v}_{kl}
        + \sigma_\mathrm{n}^2 \sum_{l\in\D_k} \mu_{kl}^2 \lVert \mathbf{v}_{kl} \rVert^2 }.
    \label{eq:dist:SINRk}
\end{equation}
\end{figure*}

\begin{proposition}[Global SINR Approximation and Optimal Weights]
\label{prop:dist:globalSINRapprox_and_optimalweighting}
    Assuming $M \ge U$ and \gls{L-MMSE} combining, the global \gls{SINR} for \gls{UE} $k$ can be approximated by:%\footnote{This approximation remains robust under pilot reuse provided that subarray selection effectively isolates high-gain links.}
    \begin{equation}
        \widehat{\Gamma}_k^\mathrm{dist} = \left( \sum_{l\in\D_k} \mu_{kl}^2 (\Gamma_{kl}^\mathrm{dist})^{-1} \right)^{-1}.
    \label{eq:dist:SINRk_from_SINRkl}
    \end{equation}
    The weights that maximize \eqref{eq:dist:SINRk_from_SINRkl} subject to $\sum\limits_{l\in\D_k} \mu_{kl} = 1$ are:
    \begin{equation}
        \mu_{kl}%^*
        =
        \begin{cases}
            \displaystyle 
            \frac{\Gamma_{kl}^\mathrm{dist}}{%\displaystyle 
            \sum_{l'\in\D_k} \Gamma_{kl'}^\mathrm{dist}}, & \text{if } l \in \D_k,\\
            0, & \text{otherwise},
        \end{cases}
    \label{eq:dist:weights:optimal}
    \end{equation}
    yielding the maximum approximate global \gls{SINR}
    \begin{equation}
        \widehat{\Gamma}_k^\mathrm{dist} = \sum_{l\in\D_k} \Gamma_{kl}^\mathrm{dist}.
    \label{eq:dist:SINRk_from_SINRkl_optweights}
    \end{equation}
\end{proposition}

\begin{proof}
    See Appendix \ref{app:proof:globalSINRapprox_and_optimalweighting}.
\end{proof}

We also refer to the optimal weighting in \eqref{eq:dist:weights:optimal} as \gls{SINR}-based weighting. A practical alternative is \gls{LSF}-based weighting, which avoids local \gls{SINR} estimation by using weights proportional to the \gls{LSF} gains:
\begin{equation}
    \mu_{kl}%^\mathrm{LSF}
    =
    \begin{cases}
         \displaystyle 
         \frac{\beta_{kl}}{%\displaystyle 
         \sum_{l'\in\D_k} \beta_{kl'}}, & \text{if } l \in \D_k,\\
        0, & \text{otherwise}.
    \end{cases}
 \label{eq:dist:weights:PCG}
\end{equation}

While \eqref{eq:dist:weights:PCG} is suboptimal as it neglects interference and channel estimation errors, it generally outperforms the equal weighting strategy, defined as:
\begin{equation}
    \mu_{kl}%^\mathrm{equal}
    =
    \begin{cases}
        %\displaystyle \frac{1}{L_k}, & \text{if } l \in \D_k,\\
        \displaystyle 1/L_k, & \text{if } l \in \D_k,\\
        0, & \text{otherwise}.
    \end{cases}
\label{eq:dist:weights:equal}
\end{equation}

The approximate \gls{SE} can be expressed via the global mapping $\hat{\eta}_k^\mathrm{dist} = g(\widehat{\Gamma}_k^\mathrm{dist})$ or, equivalently, via the component-wise mapping $\hat{\eta}_k^\mathrm{dist} = \hat{g}\big(\{ \Gamma_{kl}^\mathrm{dist} \}_{l\in\D_k}\big)$, defined as:
\begin{equation}
    \hat{g}\left(\{\Gamma_{kl}\}_{l\in\D_k}\right) = \left(1 - \frac{\tau_\mathrm{p}}{\tau_\mathrm{c}}\right) \log_2 \left(1 + \frac{1}{\sum_{l\in\D_k} \mu_{kl}^2 \Gamma_{kl}^{-1}} \right).
    %\hat{g}\left(\{\Gamma_{kl}\}_{l\in\D_k}\right) = \left(1 - \frac{\tau_\mathrm{p}}{\tau_\mathrm{c}}\right) \log_2 \left[ 1 + \frac{1}{\displaystyle \sum_{l\in\D_k} \mu_{kl}^2 \Gamma_{kl}^{-1}} \right].
\label{eq:g^}
\end{equation}

%----------------------------------
\subsection{Performance Comparison%: Centralized vs. Distributed
}
\label{subsec:system_model:results}
%----------------------------------

This subsection evaluates the uplink performance of centralized and distributed \gls{XL-MIMO} architectures as a function of the number of serving subarrays, for different subarray selection and weighting strategies.

The \gls{BS} comprises $L=16$ linearly arranged subarrays, each equipped with a $4 \times 4$ \gls{UPA} ($M=16$), serving $U=16$ \glspl{UE} in a $100 \times 50$\,m coverage area% located in front of the array
. The coherence block length is $\tau_\mathrm{c} = 16$ and $\tau_\mathrm{p} = 8$ orthogonal pilots are randomly assigned to the \glspl{UE}, which implies pilot contamination since $U > \tau_\mathrm{p}$. Table \ref{tab:parameters} summarizes the remaining parameter values adopted in this validation.
%The detailed channel model and the remaining parameters are summarized in Appendix \ref{app:channel_model} and Table \ref{tab:parameters}, respectively.

\begin{table}[!ht]
\caption{Simulation Parameters}
\label{tab:parameters}
\centering
\begin{tabular}{l c}
\toprule
\textbf{Parameter} & \textbf{Value} \\ \midrule
%Number of scheduled UEs ($U$) & 16\\
%Subarray layout ($L_\mathrm{y} \times L_\mathrm{z}$) & $16 \times 1$ \\
%Antennas per subarray ($M_\mathrm{y} \times M_\mathrm{z}$) & $4 \times 4$ \\
Number of subarrays ($L$) & 16\\
UE transmit power ($p_k$) & 20 dBm\\
Noise power ($\sigma_\mathrm{n}^2$) & $-96$ dBm\\
Channel coherence block length ($\tau_\mathrm{c}$) & 16 \\
Carrier frequency & 6 GHz \\
Antenna spacing (half-wavelength) & 2.5 cm\\
Subarray spacing & 6.0 m \\
%Angular spread ($\sigma_\varphi, \sigma_\theta$) & $10^\circ, 10^\circ$ \\
Array height & 3 m \\
UE height & 1.5 m\\
%NLoS path-loss attenuation exponent: $\gamma$ & 4\\
%Shadowing standard deviation: $\sigma_\mathrm{s}^2$ & 3 dB (LoS)\\ & 4 dB (NLoS)\\
%Average channel gain at a distance of 1 m: $\beta_0$ & $8.9125\cdot10^{-4}$\\
%Array length & 100 m\\
%Size of the square cell & 100 m \\
%Wavelength ($\lambda$) & 5 cm \\
%Number of channel realizations & 10,000 \\
%Number of subarrays serving each \gls{UE}: $L_k$ & 2 \\
\bottomrule
\end{tabular}
\end{table}
 
% We define 12 scenarios (Table~\ref{tab:scenarios}) combining different operation modes, subarray selection criteria, and weighting strategies.\footnote{Weighting is not specified for the centralized operation because it is implicit in the joint processing.}

% \begin{table}[!htbp]
%     \caption{Evaluated Scenarios.}
%     \label{tab:scenarios}
%     \centering
%     \begin{tabular}{c l l l}
%         \toprule
%         \textbf{Scenario} & \textbf{Operation} & \textbf{Selection} & \textbf{Weighting} \\ \midrule
%         1 & Centralized & Random & -- \\
%         2 & Centralized & LSF-based & -- \\
%         3 & Centralized & SINR-based & -- \\ \midrule
%         4 & Distributed & Random & Equal \\ 
%         5 & Distributed & LSF-based & Equal \\ 
%         6 & Distributed & SINR-based & Equal \\ 
%         7 & Distributed & Random & LSF-based \\ 
%         8 & Distributed & LSF-based & LSF-based \\ 
%         9 & Distributed & SINR-based & LSF-based \\ 
%         10 & Distributed & Random & SINR-based \\ 
%         11 & Distributed & LSF-based & SINR-based \\ 
%         12 & Distributed & SINR-based & SINR-based \\
%         \bottomrule
%     \end{tabular}
% \end{table}

Figure \ref{fig:mean_SE_distributed} illustrates the impact of subarray selection and weighting on the mean \gls{SE} in distributed operation. Two distinct convergence regimes are observed at the boundaries of the $L_k$ range. At $L_k=1$, performance is dictated solely by the \textit{selection} strategy, clustering the curves into three groups; weighting is irrelevant here, as the final signal estimate coincides with the local estimate from the selected subarray. Conversely, at $L_k=L=16$, the selection strategy is irrelevant, and the curves cluster into three groups determined exclusively by the \textit{weighting} scheme.

\begin{figure}[!ht]
\centering
\subfloat[Distributed operation: comparing subarray selection and weighting strategies.]{\includegraphics[trim={0mm 0mm 0mm 0mm},clip,width=0.75\linewidth]{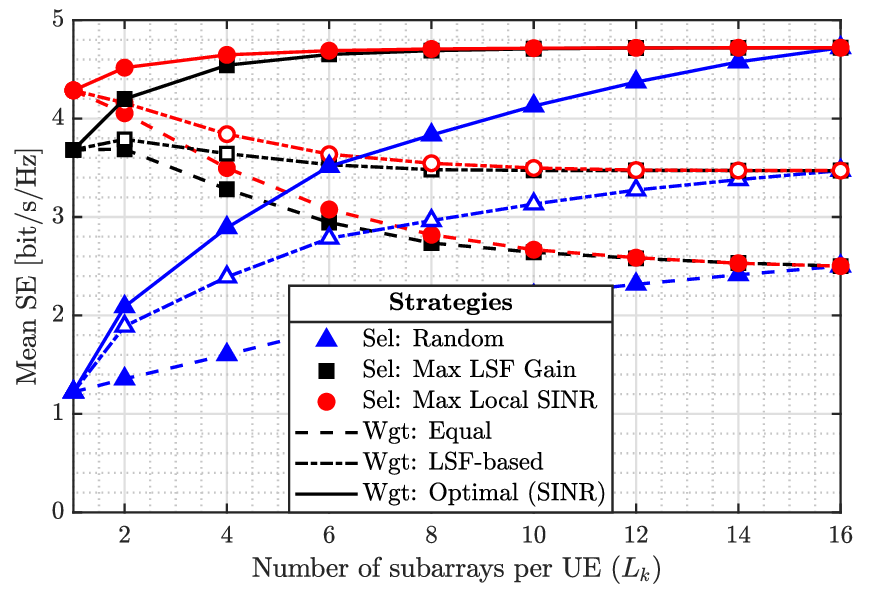}\label{fig:mean_SE_distributed}}
\\
\subfloat[Architecture comparison: centralized vs. distributed with optimal weighting, varying subarray selection strategy.]{\includegraphics[trim={0mm 0mm 0mm 0mm},clip,width=0.75\linewidth]{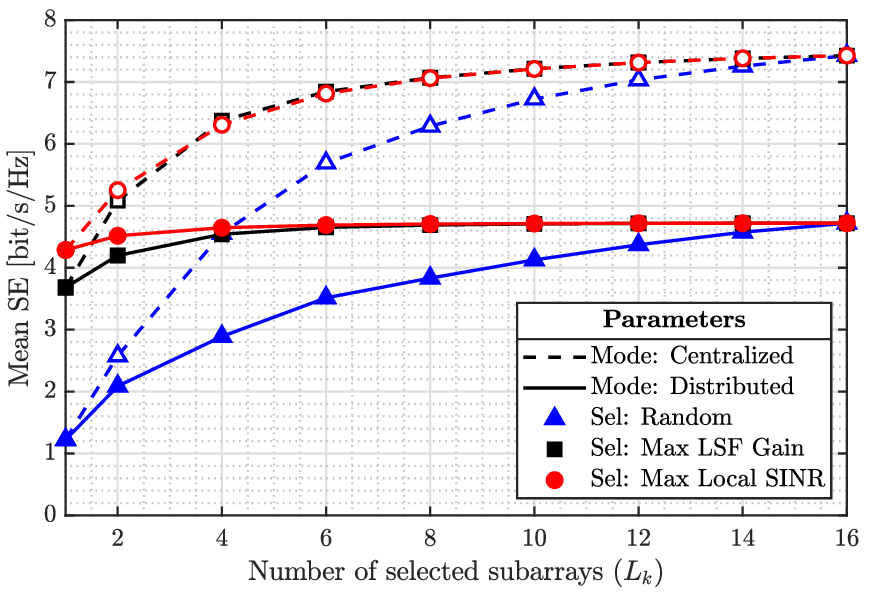}\label{fig:mean_SE}}
\caption{Mean SE versus number of selected subarrays $L_k$.% $U=16$, $\tau_\mathrm{c} = 16$, $\tau_\mathrm{p} = 8$, $M_\mathrm{y} = M_\mathrm{z} = 4$, $M = 16$, $L = 16$, $L_\mathrm{y} = 16$, $L_\mathrm{z} = 1$, cell size 100 x 100 m, MCL = 10, MCS = 1.
}
\label{fig:UL_data_estimation}
\end{figure}

%\subsubsection{Impact of Subarray Selection}
For small numbers of selected subarrays (e.g., $L_k = 2$), intelligent selection (\gls{LSF}- and \gls{SINR}-based) significantly outperforms random selection. The \gls{SINR}-based scheme yields the highest performance, confirming that interference awareness is especially critical when $L_k$ is low. As $L_k$ increases, the \gls{LSF}-based selection converges to the \gls{SINR}-based scheme in performance.
% As $L_k$ increases, the performance of \gls{LSF}-based selection converges to that of the \gls{SINR}-based scheme, suggesting that maximizing channel gain becomes an accurate proxy for maximizing \gls{SINR} as the number of contributing subarrays grows.
% This suggests that the subarrays with the strongest channel gains are the primary contributors to the signal power; thus, maximizing channel gain serves as an accurate proxy for maximizing \gls{SINR} when $L_k$ is large.

%\subsubsection{Effectiveness of Weighting Strategies}
%The combining weights are critical for scaling performance. 
The optimal (\gls{SINR}-based) weighting provides a clear upper bound, increasing monotonically with $L_k$ up to about 4.7 bit/s/Hz. 
In contrast, \textit{equal} weighting degrades performance for intelligent selection schemes as $L_k$ grows, since it does not attenuate the contribution of the subarrays with weak channels or strong interference. Consequently, increasing $L_k$ in distributed operation is only beneficial if optimal weights are applied%; otherwise, selecting a small subset of the best subarrays is preferable% to using all connections
. The \gls{LSF}-based weighting offers an intermediate solution, alleviating performance loss but still falling short of the optimum as it fails to penalize interference.

Figure \ref{fig:mean_SE} compares the distributed architecture (with optimal weighting) against the fully centralized implementation. As expected, centralized processing consistently outperforms the distributed approach due to its ability to jointly suppress inter-user interference.
The \textit{SINR-based} selection yields the highest performance for both architectures. %However, while the centralized SE (dashed red) saturates at $\approx 7.5$~bit/s/Hz, the distributed SE (solid red) saturates at a significantly lower $\approx 4.7$~bit/s/Hz, highlighting the inherent interference management limitations of local processing.
The \textit{LSF-based} scheme closely follows the SINR-based trend, reinforcing that LSF is a robust selection metric. %Finally, while \textit{random} selection performs poorly at low $L_k$, the gap narrows significantly in the centralized case as diversity increases, suggesting that selection becomes less critical in centralized setups with sufficient antennas.

The results highlight the trade-off between complexity and \gls{SE}. While centralized processing yields superior \gls{SE}, the distributed architecture offers a viable low-complexity alternative. To maximize its efficacy, two design choices are critical: \textit{a)} \gls{LSF}-based subarray selection is recommended for achieving near-optimal performance without requiring local \gls{SINR} estimation; and \textit{b)} optimal weighting is mandatory to prevent performance degradation when scaling $L_k$.

Finally, results demonstrated that selecting only a small subset of subarrays to serve each \gls{UE} is sufficient to achieve high mean \gls{SE} while significantly reducing computational complexity and fronthaul load compared to utilizing the full antenna array ($L_k=L$).

% =================== \input{SINR} % =================== 
\section{Uplink SINR Approximations}
\label{sec:SINR}

This section derives deterministic asymptotic and ergodic \gls{SINR} approximations for centralized and distributed uplink operation%, denoted by $\Gamma_k^\mathrm{cent,erg}$, $\Gamma_k^\mathrm{dist,erg}$, $\Gamma_k^\mathrm{cent,asy}$, and $\Gamma_k^\mathrm{dist,asy}$
. The asymptotic expressions describe the deterministic convergence of the instantaneous \gls{SINR} as the number of antennas grows large. Conversely, the ergodic expressions approximate the average \gls{SINR} and are particularly accurate in crowded scenarios with high user loads.

\begin{theorem}[Ergodic SINR]
\label{theo:SINR:ergodic}
    The ergodic global uplink SINR for centralized operation with MMSE combiner and the ergodic local uplink SINR for distributed operation with L-MMSE combiner are:
    \begin{subequations}
        \begin{align}
            \Gamma_k^\mathrm{cent,erg} &\approx p_k \tr\left[ \mathbf{D}_k \overline{\mathbf{Z}}_k^{-1} \mathbf{D}_k \left( \mathbf{Q}_k - \mathbf{C}_k \right) \mathbf{D}_k \right],
        \label{eq:cent:SINRk:erg} \\
            \Gamma_{kl}^\mathrm{dist,erg} &\approx p_k \tr\left[ \overline{\mathbf{Z}}_{kl}^{-1} \left( \mathbf{Q}_{kl} - \mathbf{C}_{kl} \right) \right],
        \label{eq:dist:SINRkl:erg}            
        \end{align}
    \end{subequations}
    where
    \begin{subequations}
        \begin{align}
            \overline{\mathbf{Z}}_k &= \mathbf{D}_k \left( \sum_{i \in \U \setminus \{k\}} p_i \mathbf{Q}_i + p_k \mathbf{C}_k\right) \mathbf{D}_k + \sigma_\mathrm{n}^2 \eye_{ML},
        \label{eq:cent:Zk_expectation} \\
            \overline{\mathbf{Z}}_{kl} &= \sum_{i \in \U \setminus \{k\}} p_i \mathbf{Q}_{il} + p_{k} \mathbf{C}_{kl} + \sigma_\mathrm{n}^{2} \eye_{M}.
        \label{eq:dist:Zkl_expectation}
        \end{align}
    \end{subequations}
    These approximations hold if the selected subarrays yield strong \gls{LoS} components or under high user load conditions (i.e., $U \gg ML_k$ for centralized or $U \gg M$ for distributed operation).
\end{theorem}

\begin{proof}
    See Appendix \ref{app:proof:SINR_ergo_dist}.
\end{proof}

\begin{corollary}%[Ergodic SINR Complexity]
\label{cor:complexity:ergodic}
    The computational complexity (in real multiplications) to evaluate %the ergodic uplink \gls{SINR} for \gls{UE} $k$ under centralized 
    \eqref{eq:cent:SINRk:erg} and %distributed 
    \eqref{eq:dist:SINRkl:erg} %operation
    across the $L_k$ serving subarrays is, respectively:
    \begin{subequations}
        \begin{align}
            \mathcal{C}_k^\mathrm{cent,erg} &= 3M^3L_k^3 + \frac{1}{2}M^2L_k^2 - \frac{3}{2}ML_k,
        \label{eq:cent:SINRk:erg:complexity} \\
            \mathcal{C}_k^\mathrm{dist,erg} &= 3M^3L_k + \frac{1}{2}M^2L_k - \frac{3}{2}ML_k.
        \label{eq:dist:SINRkl:erg:complexity}
        \end{align}
    \end{subequations}
\end{corollary}

\begin{proof}
    See Appendix \ref{app:dist:erglocalSINR:complexity}.
\end{proof}

\begin{theorem}[Asymptotic SINR]
\label{theo:SINR:asymptotic}
    Assume MMSE channel estimation and user loads satisfying $U \le ML_k$ for centralized or $U \le M$ for distributed operation. Furthermore, assume MMSE (centralized) or L-MMSE (distributed) combining is used, and that \textit{a)} there is no pilot reuse or \textit{b)} the channel estimates are relatively accurate and the \gls{NLoS} channel component is spatially uncorrelated.
    As the number of antennas grows large ($ML_k\to\infty$ for centralized or $M\to\infty$ for distributed operation), the instantaneous \gls{SINR} converges deterministically to the asymptotic limits given by:
    \begin{subequations}
        \begin{align}
            \Gamma_k^\mathrm{cent,asy} 
            &= \frac{\displaystyle \tr(\mathbf{D}_k \mathbf{X}_k \mathbf{D}_k) - \sum_{\substack{i\in\U \setminus \{k\}}} \frac{\tr(\mathbf{D}_k \mathbf{X}_i \mathbf{D}_k \mathbf{X}_k \mathbf{D}_k)}{\tr(\mathbf{D}_k \mathbf{X}_i \mathbf{D}_k)}}{\displaystyle \sum_{i\in\U} \frac{p_i}{ML_kp_k} \tr(\mathbf{D}_k \mathbf{C}_i \mathbf{D}_k) + \sigma_\mathrm{n}^2 },
        \label{eq:cent:SINRk:asy} \\
            \Gamma_{kl}^\mathrm{dist,asy} &= \frac{\displaystyle \tr\left( \mathbf{X}_{kl} \right) - \sum_{i\in\U \setminus \{k\}} \frac{\tr(\mathbf{X}_{kl} \mathbf{X}_{il})}{\tr(\mathbf{X}_{il})} }{\displaystyle \sum_{i\in\U} \frac{p_i}{Mp_k} \tr(\mathbf{C}_{il}) + \sigma_\mathrm{n}^2},
        \label{eq:dist:SINRkl:asy}
        \end{align}
    \end{subequations}
    where $\mathbf{X}_k = \mathbf{Q}_k - \mathbf{C}_k$ and $\mathbf{X}_{kl} = \mathbf{Q}_{kl} - \mathbf{C}_{kl}$.
\end{theorem}

\begin{proof}
     See Appendix \ref{app:proof:SINR_asym_dist}.
\end{proof}

\begin{corollary}
\label{cor:complexity:asymptotic}
    The computational complexity (in real multiplications) to evaluate \eqref{eq:cent:SINRk:asy} and \eqref{eq:dist:SINRkl:asy} across the $L_k$ serving subarrays is, respectively:
    \begin{subequations}
        \begin{align}
            \mathcal{C}_k^\mathrm{cent,asy} &= 3 (U-1) M^2 L_k^2,
        \label{eq:cent:SINRk:asy:complexity} \\
            \mathcal{C}_k^\mathrm{dist,asy} &= 3 (U-1) M^2 L_k.
        \label{eq:dist:SINRkl:asy:complexity}
        \end{align}
    \end{subequations}
\end{corollary}
    
\begin{proof}
    %See Appendix \ref{app:dist:asylocalSINR:complexity}.
    Complexity is dominated by evaluating trace terms $\tr(\mathbf{A}\mathbf{B})$ for the $U-1$ interfering UEs ($i \neq k$). Computing the trace of a product of two square matrices of size $N$ requires calculating only the diagonal elements of the product, costing $N^2$ complex multiplications.
    For distributed operation \eqref{eq:dist:SINRkl:asy}, this involves $M \times M$ matrices ($N=M$) for each of the $L_k$ subarrays.
    For centralized operation \eqref{eq:cent:SINRk:asy}, the matrices have effective size $ML_k \times ML_k$ ($N=ML_k$).
    Assuming three real multiplications per complex multiplication \cite[p.~378]{massivemimobook}, we arrive at \eqref{eq:cent:SINRk:asy:complexity} and \eqref{eq:dist:SINRkl:asy:complexity}.
\end{proof}

The approximate achievable \gls{SE} is obtained via the mapping $g(\cdot)$ in \eqref{eq:g}. For centralized operation, we evaluate $\eta_k^\mathrm{cent,asy} = g(\Gamma_k^\mathrm{cent,asy})$ and $\eta_k^\mathrm{cent,erg} = g(\Gamma_k^\mathrm{cent,erg})$. For distributed operation, it is reconstructed from the local \gls{SINR} approximations as $\hat{\eta}_k^\mathrm{dist,asy} = \hat{g}\big(\{ \Gamma_{kl}^\mathrm{dist,asy} \}_{l\in\D_k}\big)$ and $\hat{\eta}_k^\mathrm{dist,erg} = \hat{g}\big(\{ \Gamma_{kl}^\mathrm{dist,erg} \}_{l\in\D_k}\big)$.

%----------------------------------
\subsection{Approximation Accuracy}
\label{subsec:SINR:accuracy}
%----------------------------------

This section validates the derived \gls{SINR} approximations. Simulations assume $L=16$ and $L_k=4$, employing \gls{LSF}-based subarray selection and optimal weighting% (for distributed operation)
.\footnote{The choice of $L_k=4$ is justified by the performance-complexity trade-off observed in Fig. \ref{fig:mean_SE}.} We enforce pilot contamination by setting $\tau_p = U/2$ with random assignment, while $\tau_c = U$. Remaining parameters are listed in Table \ref{tab:parameters}.

We assess the accuracy of the derived approximations using the \gls{MNAE} metric. For a generic variable $x$ and its estimate $\hat{x}$, the \gls{MNAE} is defined as:\footnote{Throughout this work, $\mathbb{E}\{\cdot\}$ denotes the expectation over \gls{SSF} realizations.}
\begin{equation}
    \nu(x,\hat{x}) = \E\left\{ \frac{\lvert \hat{x} - x \rvert}{x} \right\}.
\label{eq:MNAE}
\end{equation}

To validate the global \gls{SINR} approximation for distributed operation derived in Proposition \ref{prop:dist:globalSINRapprox_and_optimalweighting}, we evaluate the error between the true distributed \gls{SE} and its estimate based on local \glspl{SINR}:
\begin{equation}
    \tilde{\eta}_k^\mathrm{dist} = \nu(\eta_k^\mathrm{dist}, \hat{\eta}_k^\mathrm{dist}).
\label{eq:MNAE:inst_dist}
\end{equation}

Since the asymptotic expressions target the \emph{instantaneous} performance, their accuracy is measured relative to the true instantaneous \gls{SE}:
\begin{equation}
    \tilde{\eta}_k^\mathrm{op,asy} = \nu(\eta_k^\mathrm{op}, \hat{\eta}_k^\mathrm{op,asy}), \quad \text{for } \mathrm{op} \in \{\mathrm{cent, dist}\}.
\label{eq:MNAE:asy}
\end{equation}

Conversely, the ergodic expressions target the \emph{average} performance. Thus, we compare them against the true mean \gls{SE}:
\begin{equation}
    \tilde{\eta}_k^\mathrm{op,erg} = \nu(\E\{\eta_k^\mathrm{op}\}, \hat{\eta}_k^\mathrm{op,erg}), \quad \text{for } \mathrm{op} \in \{\mathrm{cent, dist}\}.
\label{eq:MNAE:erg}
\end{equation}

We first assess the accuracy of the global \gls{SINR} approximation for distributed operation in \eqref{eq:dist:SINRk_from_SINRkl} using Fig. \ref{fig:MNAE_instSE_x_M_dist}, which plots the \gls{MNAE} $\tilde{\eta}_k^\mathrm{dist}$ versus $M$.\footnote{We actually plot the accuracy averaged across both scheduled \glspl{UE} and channel realizations.} The approximation error decreases steadily with $M$ and increases with $U$, since a larger $M/U$ ratio provides stronger interference suppression, satisfying the derivation's assumption of negligible coherent interference. On the other hand, selecting fewer subarrays (e.g., $L_k=2$), yields lower \gls{MNAE} than larger subsets (e.g., $L_k=8$), as the \gls{LSF}-based strategy can restrict selection to t{he strongest subarrays where $\mathbf{v}_{kl}^\mathrm{H} \widehat{\mathbf{h}}_{kl} \approx 1$.\footnote{See Appendix \ref{app:proof:globalSINRapprox_and_optimalweighting} for more details on the validity conditions for the approximation.} These factors render the approximation robust even in the overloaded regime ($M < U$).% Conversely, increasing the number of scheduled \glspl{UE} ($U$) naturally degrades accuracy since the channel estimates deteriorate.

\begin{figure}[!ht]
    \centering
    \includegraphics[trim={0mm 0mm 0mm 0mm},clip,width=.75\linewidth]{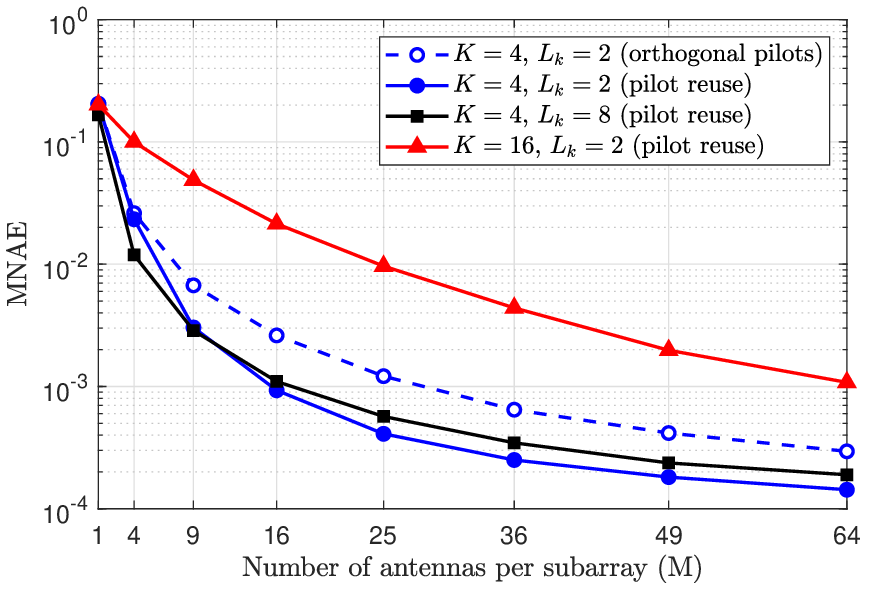}
    \caption{MNAE of the global SE approximation for distributed operation versus subarray antennas.}
\label{fig:MNAE_instSE_x_M_dist}
\end{figure}

Figure \ref{fig:MNAE_SE_vs_U} evaluates the ergodic and asymptotic approximations as a function of $U$, with fixed $M=16$ and $L_k=4$. The subfigures depict the \gls{MNAE} for distributed (Fig.~\ref{fig:MNAE_SE_x_U_dist}) and centralized (Fig.~\ref{fig:MNAE_SE_x_U_cent}) operations, assessing the impact of spatial correlation and pilot contamination. To this end, we compare spatially correlated vs. uncorrelated fading, and contrast orthogonal pilot assignment ($\tau_\mathrm{p} = U$) with pilot reuse ($\tau_\mathrm{p} = U/2$, with random assignment).

\begin{figure}[!ht]
    \centering
    \subfloat[Distributed operation.]{\includegraphics[trim={0mm 0mm 0mm 0mm},clip,width=.75\linewidth]{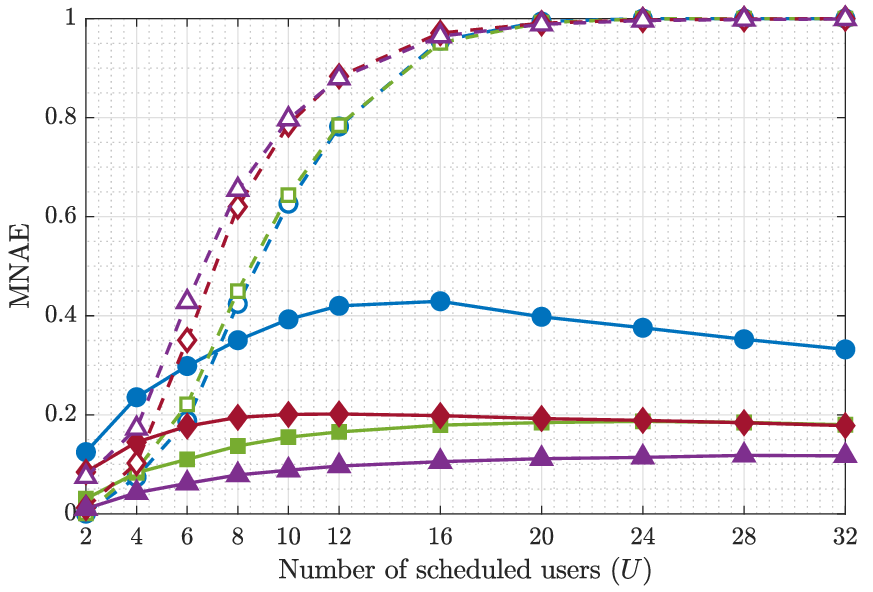}
    \label{fig:MNAE_SE_x_U_dist}}
    \\
    \subfloat[Centralized operation.]{\includegraphics[trim={0mm 0mm 0mm 0mm},clip,width=.75\linewidth]{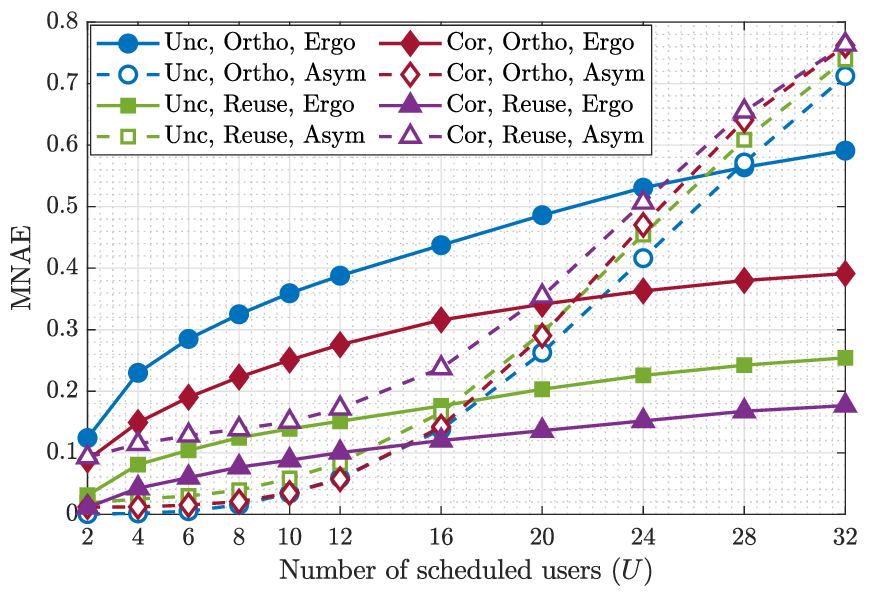}
    \label{fig:MNAE_SE_x_U_cent}}
    \caption{MNAE of the deterministic SE approximations vs. number of UEs, contrasting ergodic and asymptotic approximations under varying spatial correlation and pilot reuse scenarios.}
\label{fig:MNAE_SE_vs_U}
\end{figure}

%The following behaviors are observed regarding their accuracy:
%\subsubsection{System loading effect}
%The MNAE of the asymptotic expression increases with $U$. As the system load grows, the variance of the aggregate interference becomes more significant, causing the deterministic limit, where interference is replaced by its mean, to become less representative of instantaneous channel realizations. This divergence is more pronounced in \textit{distributed operation} compared to \textit{centralized operation}, as the latter benefits from a larger total number of antennas ($ML_k$) involved in the combining process, which accelerates the convergence to the deterministic limit, effectively delaying the error growth.
The \gls{MNAE} of the asymptotic expression increases with $U$. As the system load grows, the variance of the aggregate interference becomes more significant, causing the instantaneous realizations to deviate further from the deterministic limit. This divergence is more pronounced in compared to the centralized case. The latter benefits from a higher diversity order ($ML_k$ antennas involved in combining), which accelerates convergence to the deterministic mean.

%\subsubsection{Channel correlation and pilot reuse}
%These figures also show that spatial correlation and pilot reuse degrade the accuracy of the asymptotic approximations, since the derivations of \eqref{eq:cent:SINRk:asy} and \eqref{eq:dist:SINRkl:asy} simplify the interference terms by assuming spatially uncorrelated NLoS components and small channel estimation errors. Furthermore, spatially correlated \gls{NLoS} reduces the effective degrees of freedom (spatial diversity), which hinders channel hardening and thereby delays the \gls{SINR} convergence to the deterministic regime. Consequently, spatial correlation introduces a direct modeling error. The presence of spatial correlation and pilot contamination (which scales with $U$) directly violates these tractability assumptions, leading to the observed divergence between the asymptotic approximation and the instantaneous \gls{SE}.
According to Fig~\ref{fig:MNAE_SE_vs_U}, spatial correlation and pilot reuse degrade the accuracy of the asymptotic approximations by violating the tractability assumptions underlying \eqref{eq:cent:SINRk:asy} and \eqref{eq:dist:SINRkl:asy}, which rely on uncorrelated \gls{NLoS} components and small channel estimation errors. Furthermore, correlation reduces spatial diversity, which hinders channel hardening and thereby delays the \gls{SINR} convergence to the deterministic regime, while pilot contamination scales with $U$. These factors introduce direct modeling discrepancies, causing the instantaneous \gls{SE} to diverge from the asymptotic approximation.

Counter-intuitively, spatial correlation and pilot reuse reduce the \gls{MNAE} of the ergodic approximations by increasing the deterministic covariance $\mathbf{C}_{kl}$, which stabilizes $\mathbf{Z}_{kl}$ around its mean, better satisfying the Neumann series condition required for these approximations. The error evolution is driven by competing effects: it initially rises with $U$ as additional interferers increase the variance of $\mathbf{Z}_{kl}$, and then it peaks near the spatial degrees of freedom ($U \approx M$ for distributed operation and $U \approx ML_k$ for centralized). As we increase $U$ beyond this peak, $\mathbf{Z}_{kl}$ concentrates around its expected value. This explains why, for the same $M$, the ergodic approximation is more accurate in distributed operation than in centralized: it reaches saturation significantly earlier than the centralized.

Given the contrasting error behaviors of the ergodic and asymptotic approximations, we define a switching point, $U_{\text{switch}}$, to select the best approximation. The asymptotic expression is prioritized for $U \le U_{\text{switch}}$ due to its superior accuracy and lower complexity, while the ergodic expression is adopted for $U > U_{\text{switch}}$. The optimal $U_\mathrm{switch}$ depends on channel conditions and the pilot assignment strategy, shifting toward lower values of $U$ when correlation or pilot reuse is present. Table \ref{tab:U_switch_values} summarizes the empirically determined switching thresholds derived from the intersection points of the asymptotic and ergodic error curves in Fig.~\ref{fig:MNAE_SE_vs_U}. Note that the thresholds for centralized operation are consistently scaled by $L_k$ compared to the distributed case. %For instance, in the scenario with uncorrelated NLoS and no pilot reuse, the threshold can be defined at half the total number of combining antennas: $U_{\text{switch}} = ML_k/2$ for centralized and $U_{\text{switch}} = M/2$ for distributed operation.
Note also that, in the most severe scenario (correlated NLoS with pilot reuse), the threshold drops to $U_{\text{switch}} \approx 0$. This indicates that the asymptotic approximation is unreliable for any practical user load under these conditions, and the ergodic approximation should be employed exclusively for all $U$.

\begin{table}[!ht]
\centering
\caption{Switching threshold $U_{\text{switch}}$ values.}
\label{tab:U_switch_values}
\setlength\tabcolsep{3pt}
\small
\begin{tabular}{llcc}
\toprule
\multicolumn{2}{c}{\textbf{Scenario}} & \multicolumn{2}{c}{\textbf{Operation}} \\
\cmidrule(r){1-2} \cmidrule(l){3-4}
\textbf{NLoS} & \textbf{Pilots} & \textbf{Centralized} & \textbf{Distributed} \\
\midrule
\multirow{2}{*}{Uncorrelated} & Orthogonal & $ML_k/2$ & $M/2$ \\
 & Reuse & $ML_k/4$ & $M/4$ \\
\cmidrule{1-4}
\multirow{2}{*}{Correlated} & Orthogonal & $ML_k/4$ & $M/4$ \\
 & Reuse & $0$ & $0$ \\
\bottomrule
\addlinespace
%\multicolumn{4}{l}{\footnotesize Note: $\tau_p = U$ for orthogonal pilots; $\tau_p = U/2$ for reuse.}
\end{tabular}
\end{table}

%In summary, the proposed switching strategy adapts to the severity of the propagation environment. As channel conditions become more challenging—characterized by spatial correlation and heavy interference—the algorithm conservatively transitions to the robust ergodic approximation at lower user loads. Consequently, the ergodic approximation proves to be more accurate in scenarios involving heavy system loading (high $U$), spatially correlated NLoS components, and pilot reuse. Furthermore, although the tightness of the ergodic SE approximation relies on interference hardening (typically associated with heavy loading), the results confirm that this expression remains accurate even when $U$ exceeds $M$ only moderately.

%Conversely, the asymptotic formulation effectively approximates the instantaneous SINR in regimes where $M$ is sufficiently large to induce channel hardening. Fundamentally, while the ergodic expression estimates the average SINR, the asymptotic expression serves as a deterministic proxy for the instantaneous performance in the large-system limit.

%This provides empirical support for the pilot assignment strategies proposed in Chapter~\ref{chp:algorithms}. The results confirm the decision rules implemented in Algorithms \ref{alg:cent:NMSEbased} and \ref{alg:dist:SINRbased}, which automatically switch between the candidate expressions so that, at each scheduling step, the estimator with the smallest error is selected.

The validity of the asymptotic approximations relies on channel hardening, confirmed by the monotonic \gls{MNAE} decrease when increasing $M$ in Fig.~\ref{fig:MNAE_asym_vs_M}. Centralized operation converges significantly faster than distributed operation, achieving accuracy even at moderate $M$ due to the larger effective array size ($ML_k$ vs. $M$) involved in joint combining. Conversely, accuracy degrades with pilot reuse and higher user loads ($U$). The simulation setup includes spatially correlated NLoS components and $L_k=2$.

\begin{figure}[!ht]
    \centering
    \includegraphics[trim={0mm 0mm 0mm 0mm},clip,width=.75\linewidth]{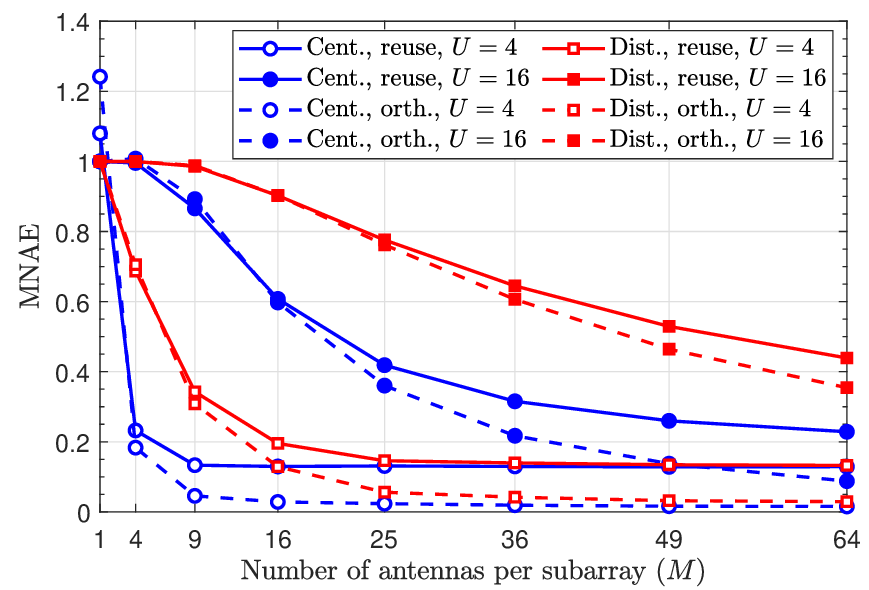}
    \caption{MNAE of the asymptotic SE approximation versus subarray antennas ($M$).}
\label{fig:MNAE_asym_vs_M}
\end{figure}

% ===================  \input{algorithms} % =================== 

\section{Applying User Scheduling, Pilot Assignment and Subarray Selection in XL-MIMO}
\label{sec:algorithms}

% %--------------------------------------------------------
% \subsection{NMSE-Based %Joint Subarray Selection, 
% User Scheduling \& Pilot Assignment
% }
% \label{subsec:alg:NMSE}
%--------------------------------------------------------

%Optimal resource allocation is combinatorial and intractable, often requiring suboptimal heuristics \cite{UL_XLMIMO_2020,CFbook}. 
A widely used benchmark for joint pilot assignment and subarray selection is the greedy algorithm in \cite[Sec. 4.4]{CFbook}, which %avoids pilot contamination through spatial separation. It 
follows a three-stage procedure: (i) assign $\tau_\mathrm{p}$ orthogonal pilots to $\tau_\mathrm{p}$ random \glspl{UE}; (ii) sequentially assign pilots to the remaining \glspl{UE} so as to minimize pilot contamination\footnote{When assigning a candidate pilot $t$ to a \gls{UE} $k$, \cite[Sec. 4.4]{CFbook} measures pilot contamination as the sum of the \gls{NLoS} channel gains of the \glspl{UE} already using pilot $t$, evaluated at \gls{UE} $k$'s strongest subarray. Only the \gls{NLoS} component is considered because the deterministic \gls{LoS} component does not contribute to the channel estimation error.}; and (iii) perform subarray selection, where each subarray locally chooses which $\tau_\mathrm{p}$ \glspl{UE} to serve, specifically one \gls{UE} per pilot (namely, the strongest \gls{UE} using that pilot). {While this approach ensures minimal computational complexity, it does not strictly guarantee full user scheduling, as certain \glspl{UE} may remain unselected by any subarray.}% This decentralized approach leverages slowly varying \gls{LSF} statistics, minimizing fronthaul overhead while ensuring robust service coverage.
% Although it does not strictly guarantee that all UEs are scheduled, in practice it almost always ensures that each user is served by at least one subarray.

This section proposes joint user scheduling and pilot assignment algorithms to enhance the minimum per-user \gls{SE}. %Algorithms \ref{alg:cent:SINRbased} and \ref{alg:dist:SINRbased}, respectively for centralized and distributed operation, base its decisions on the \gls{SINR} approximations derived in Section \ref{sec:SINR}, while Algorithm \ref{alg:NMSEbased} uses the channel estimation \gls{NMSE}.
Algorithm \ref{alg:NMSEbased} maximizes the number of scheduled \glspl{UE} subject to a maximum channel estimation error threshold, $\gamma_\mathrm{th}$ (line 2). %Taking channel statistics and subarray selection sets as inputs, it  enhances \gls{SE} by ensuring accurate channel estimates%, which is critical in dense deployment scenarios. After initializing the scheduled set $\U = \emptyset$ and the number of pilot sequences $\tau_\mathrm{p} = 0$ (line 1), 
It iteratively schedules \glspl{UE} in descending order of \gls{LSF} gain %until the maximum \gls{NMSE} exceeds $\gamma_\mathrm{th}$ (line 3). %In each iteration, a candidate \gls{UE} $k$ is added to a tentative set $\tilde{\U} = \U \cup \{k\}$
(line 3). In each step, the candidate \gls{UE} is assigned the pilot (either an existing or a new one) that minimizes the maximum channel estimation \gls{NMSE} across the candidate user set $\tilde{\U}$ (lines 4-9). %For each candidate pilot, it updates the channel estimation \glspl{NMSE} for each \gls{UE} $i \in \tilde{\U}$, and \gls{UE} $k$ is then assigned the pilot that minimizes the maximum \gls{NMSE} across $\tilde{\U}$ (lines 6-9). 
If the resulting maximum \gls{NMSE} satisfies $\gamma_\mathrm{th}$, the \gls{UE} is scheduled (lines 10-12).%, and the pilot count $\tau_\mathrm{p}$ is updated if a new sequence was used (line 11).\footnote{To maintain efficiency, covariance updates within the inner loop are restricted to the specific pilot $t$ being evaluated (unless the pilot dimension changes), with permanent updates applied only after the optimal assignment is finalized.}
%\footnote{In line 6, the algorithm only needs to update $\mathbf{\Psi}_{\tilde{t}l}$ for $\tilde{t} = t$, rather than recomputing it for all $\tilde{t} \in {1, \dots, \tau_\mathrm{p}'}$, except in the case where $\tau_\mathrm{p}'$ increases during the current iteration, \textit{i.e.}, when $t = \tau_\mathrm{p} + 1$. Similarly, the algorithm only needs to recompute $\mathbf{C}_{il}$ and $\gamma_i$ for the \glspl{UE} assigned to pilot $t$, rather than for all $i \in \tilde{\U}$, unless a new pilot has been added in this iteration. These selective updates ensure computational efficiency and are only valid for the current trial assignment. Only the updates corresponding to the pilot ultimately selected for user $k$ will be retained and carried over to the next iteration of the while loop.}

\begin{algorithm}[!ht]
\caption{NMSE-Based %Joint Subarray Selection, 
User Scheduling \& Pilot Assignment
}
\label{alg:NMSEbased}
\begin{algorithmic}[1]
    \Require
    \Statex NMSE threshold: $\gamma_{\mathrm{th}}$
    \Statex Subarray selection sets: $\{\D_i\}_{i\in\K}$
    \Statex Channel statistics: $\{\mathbf{Q}_{kl},\mathbf{R}_{kl}\}_{k\in\K,l\in\D}$, $\{\beta_k\}_{k\in\K}$
    \Statex Transmit powers: $\{p_k\}_{k\in\K}$
    %\Statex Noise variance: $\sigma_\mathrm{n}^2$
    \Ensure
    \Statex Set of scheduled users: \(\U\)
    \Statex Pilot indices: $\{t_k\}_{k\in\U}$
    %\vspace{0.5em}
    \State Initialize: $\U\gets\emptyset$,\; $\tau_\mathrm{p}\gets0$,\; $\gamma_{\max}\gets0$
    \While{$|\U|<K$ \textbf{and} $\gamma_{\max}\le\gamma_\mathrm{th}$}
        \State $k\gets%\displaystyle 
        \operatorname*{arg\,max}_{i\in\K\setminus\U} \beta_i$ and $\tilde{\U} \gets \U \cup \{k\}$
        \For{$t=1$ \textbf{to} $\tau_\mathrm{p}+1$}
            \State $t_k \gets t$,\; $\tau%_\mathrm{p}' 
            \gets \max(\tau_\mathrm{p},t_k)$,\; $\zeta \gets \{1,\ldots,\tau\}$
            \State $\mathbf{\Psi}_{\tilde{t}l} \gets %\displaystyle 
            \sum_{\substack{i\in\tilde{\U} \\ t_i=\tilde{t}}} p_i \tau%_\mathrm{p}' 
            \mathbf{R}_{il} + \sigma_\mathrm{n}^2 \eye_M,\quad l\in\D,\; \tilde{t}\in\zeta$
            \State $\gamma_{\max}^{(t)} \gets \underset{i\in\tilde{\U}}{\max} \frac{ %\displaystyle 
            \sum_{l\in\D_i} \tr(\mathbf{R}_{il} - p_i \tau%_\mathrm{p}' 
            \mathbf{R}_{il} \mathbf{\Psi}_{t_il}^{-1} \mathbf{R}_{il})}{%\displaystyle 
            \sum_{l\in\D_i} \tr(\mathbf{Q}_{il})}$
        \EndFor
        \State $t_k \gets \displaystyle \operatorname*{arg\,min}_{t\in\{1,\ldots,\tau_\mathrm{p}+1\}} \gamma_{\max}^{(t)}$ and $\gamma_{\max} \gets \gamma_{\max}^{(t_k)}$
        \If{\(\gamma_{\max} \leq \gamma_\mathrm{th}\)}
            \State $\U \gets \tilde{\U}$ and $\tau_\mathrm{p} \gets \max(\tau_\mathrm{p}, t_k)$
        \EndIf
    \EndWhile
\end{algorithmic}
\end{algorithm}

% %----------------------------------------------------
% \subsection{SINR-Based User Scheduling \& Pilot Assignment}
% \label{subsec:algorithms:SINR_based}
% %----------------------------------------------------

Designed for centralized and distributed operation, respectively, Algorithms \ref{alg:cent:SINRbased} and \ref{alg:dist:SINRbased} adopt the same framework of Algorithm \ref{alg:NMSEbased}. However, rather than minimizing the maximum channel estimation \gls{NMSE}, their pilot assignment criteria maximizes the minimum per-user \gls{SE}, and their user scheduling strategy maximizes the number of scheduled \glspl{UE} subject to a minimum per-user \gls{SE}, $\eta_\mathrm{th}$. More than scheduling \glspl{UE} and assigning pilots, Algorithm \ref{alg:dist:SINRbased} computes the optimal weights used to obtain the final signal estimates from the local estimates (lines 15 and 21).% Before executing the algorithm, a subarray selection step is required, in which each UE $k$ is associated with the $L_k$ subarrays exhibiting the highest \gls{LSF} gains to this UE.
% In other words, they iteratively schedules \glspl{UE} while the minimum per-user \gls{SE} is above the per-user \gls{SE} threshold.

%To ensure estimation accuracy across loading regimes, both algorithms employ a conditional switching rule (see Sec. \ref{subsec:SINR:accuracy}): asymptotic SINR approximations (\eqref{eq:cent:SINRk:asy}, \eqref{eq:dist:SINRkl:asy}) are utilized for low loads ($\lvert\tilde{\U}\rvert \leq U_\mathrm{switch}$), while ergodic approximations (\eqref{eq:cent:SINRk:erg}, \eqref{eq:dist:SINRkl:erg}) are applied for high loads ($\lvert\tilde{\U}\rvert > U_\mathrm{switch}$). Additionally, Algorithm \ref{alg:dist:SINRbased} computes optimal combining weights, approximating the global \gls{SINR} as the sum of local \glspl{SINR} following Proposition \ref{prop:dist:globalSINRapprox_and_optimalweighting}.

\begin{algorithm}[!ht]
\caption{SINR-Based Joint User Scheduling \& Pilot Assignment For Centralized Operation}
\label{alg:cent:SINRbased}
\begin{algorithmic}[1]
    \Require
        \Statex Switching point: $U_\mathrm{switch}$
        \Statex SE threshold: $\eta_\mathrm{th}$
        \Statex Subarray selection matrices: $\{\mathbf{D}_k\}_{k\in\K}$
        \Statex Number of subarrays serving each user: $\{L_k\}_{k\in\K}$
        \Statex Channel statistics: $\{\mathbf{Q}_{k},\mathbf{R}_{k},\beta_k\}_{k\in\K}$
        \Statex Transmit powers: $\{p_k\}_{k\in\K}$
        %\Statex Noise variance: $\sigma_\mathrm{n}^2$
    \Ensure
        \Statex Set of scheduled users: \(\U\)
        \Statex Pilot indices: $\{t_k\}_{k\in\U}$
        %\Statex Number of pilot sequences: $\tau_\mathrm{p}$
    %\vspace{0.5em}
    \State Initialize: $\U\gets\emptyset$,\; $\tau_\mathrm{p}\gets0$,\; $\eta_{\min} \gets 2 \eta_\mathrm{th}$
    \While{$|\U|<K$ \textbf{and} $\eta_{\min}>\eta_\mathrm{th}$}
        \State $k\gets%\displaystyle 
        \operatorname*{arg\,max}_{i\in\K\setminus\U} \beta_i$ and $\tilde{\U} \gets \U \cup \{k\}$
        \For{$t=1$ \textbf{to} $\tau_\mathrm{p}+1$}
            \State $t_k \gets t$,\; $\tau \gets \max(\tau_\mathrm{p},\,t_k)$,\; $\zeta \gets \{1,\ldots,\tau\}$
            \State $\mathbf{\Psi}_{\tilde{t}} \gets %\displaystyle 
            \sum_{\substack{i\in\tilde{\U} \\ t_i=\tilde{t}}} p_i \tau%_\mathrm{p}' 
            \mathbf{R}_i + \sigma_\mathrm{n}^2 \eye_{ML},\quad \tilde{t}\in\zeta$
            \State $\mathbf{C}_i \gets \mathbf{R}_i - p_i \tau \mathbf{R}_i \mathbf{\Psi}_{t_i}^{-1} \mathbf{R}_i,\quad i\in\tilde{\U}$
            %\vspace{1em}
            \If{$|\tilde{\U}| > U_\mathrm{switch}$}: for $i\in\tilde{\U}$:
                \State $\overline{\mathbf{Z}}_i \gets \mathbf{D}_i \Bigl( \sum\limits_{\substack{i'\in\tilde{\U}\\i'\neq i}} p_{i'} \mathbf{Q}_{i'} + p_i \mathbf{C}_i\Bigr) \mathbf{D}_i + \sigma_\mathrm{n}^2 \eye_{ML}$
                \State %\qquad\qquad 
                $\Gamma_i \gets p_i \tr[\mathbf{D}_i \overline{\mathbf{Z}}_i^{-1} \mathbf{D}_i \left(\mathbf{Q}_i - \mathbf{C}_i \right) \mathbf{D}_i]$
            \Else: for $i\in\tilde{\U}$:
                \State $\mathbf{S}_{ii'} \gets \mathbf{D}_i (\mathbf{Q}_{i'} - \mathbf{C}_{i'}) \mathbf{D}_i$,\; $i' \in \tilde{\U}$
                \State %\qquad\qquad 
                $\Gamma_i \gets \frac{\displaystyle 
                \tr(\mathbf{S}_{ii}) - \sum_{\substack{i'\in\tilde{\U}\setminus\{i\}}} \frac{\tr(\mathbf{S}_{ii} \mathbf{S}_{ii'})}{\tr(\mathbf{S}_{ii'})}}{\displaystyle \sum_{i'\in\tilde{\U}} \frac{p_{i'}}{ML_ip_i} \tr(\mathbf{D}_i\mathbf{C}_{i'}\mathbf{D}_i) + \sigma_\mathrm{n}^2}$
            \EndIf
            \State $\eta_{\min}^{(t)} \gets \displaystyle \min_{i\in\tilde{\U}} \left( 1-\frac{\tau}{\tau_\mathrm{c}} \right) \log_2 (1+\Gamma_i)$
        \EndFor
        \State $t_k \gets \displaystyle \operatorname*{arg\,max}_{t\in\{1,\ldots,\tau_\mathrm{p}+1\}} \eta_{\min}^{(t)}$ and $\eta_{\min} \gets \eta_{\min}^{(t_k)}$
        \If{$\eta_{\min} \geq \eta_\mathrm{th}$}
            \State $\U \gets \tilde{\U}$ and $\tau_\mathrm{p} \gets \max(\tau_\mathrm{p}, t_k)$
        \EndIf
    \EndWhile
\end{algorithmic}
\end{algorithm}

\begin{algorithm}[!ht]
\caption{SINR-Based Joint %Subarray Selection, 
User Scheduling \& Pilot Assignment For Distributed Operation}
\label{alg:dist:SINRbased}
\begin{algorithmic}[1]
    \Require 
    \Statex Switching point: $U_\mathrm{switch}$
    \Statex SE threshold: $\eta_\mathrm{th}$
    %\Statex Number of subarrays serving each user: $\{L_i\}_{i\in\K}$
    \Statex Subarray selection sets: $\{\D_i\}_{i\in\K}$
    \Statex Channel statistics: $\{\mathbf{Q}_{kl},\mathbf{R}_{kl}\}_{k\in\K,\,l\in\D}$, $\{\beta_k\}_{k\in\K}$
    \Statex Transmit powers: $\{p_k\}_{k\in\K}$
    %\Statex Noise variance: $\sigma_\mathrm{n}^2$
    \Ensure
    \Statex Set of scheduled users: \(\U\)
    \Statex Pilot indices: $\{t_k\}_{k\in\U}$
    \Statex Weights: $\{\mu_{kl}\}_{k\in\U,\,l\in\D_k}$ 
    %\vspace{0.5em}
    \State Initialize: $\U\gets\emptyset$,\, $\tau_\mathrm{p}\gets0$,\, $\eta_{\min} \gets 2 \eta_\mathrm{th}$
    \While{$\lvert\U\rvert<K$ \textbf{and} $\eta_{\min} > \eta_\mathrm{th}$}
        \State $k\gets%\displaystyle 
        \operatorname*{arg\,max}_{i\in\K\setminus\U} \beta_i$ and $\tilde{\U} \gets \U \cup \{k\}$
        \For{$t=1$ \textbf{to} $\tau_\mathrm{p}+1$}
            \State $t_k \gets t$,\; $\tau \gets \max(\tau_\mathrm{p},\,t_k)$,\; $\zeta \gets \{1,\ldots,\tau\}$
            \State $\mathbf{\Psi}_{\tilde{t}l} \gets %\displaystyle 
            \sum_{\substack{i\in\tilde{\U} \\ t_i=\tilde{t}}} p_i \tau%_\mathrm{p}' 
            \mathbf{R}_{il} + \sigma_\mathrm{n}^2 \eye_M,\quad l\in\D,\; \tilde{t}\in\zeta$
            \State \( \mathbf{C}_{il} \gets \mathbf{R}_{il} - p_i \tau \mathbf{R}_{il} \mathbf{\Psi}_{t_il}^{-1} \mathbf{R}_{il},\quad l\in\D,\; i\in\tilde{\U}\)
            \State $\mathbf{S}_{il} = \mathbf{Q}_{il} - \mathbf{C}_{il},\quad l\in\D,\; i\in\tilde{\U}$
            \If{$|\tilde{\U}| > U_\mathrm{switch}$}: for $i\in\tilde{\U}$ and $l\in\D_i$:
                \State $\overline{\mathbf{Z}}_{il} \gets \displaystyle \sum_{i'\in\tilde{\U} \setminus\{i\}}
                p_{i'} \mathbf{Q}_{i'l} + p_i \mathbf{C}_{il} + \sigma_\mathrm{n}^{2} \eye_M$
                \State $\Gamma_{il} \gets \displaystyle p_i \tr[\overline{\mathbf{Z}}_{il}^{-1}(\mathbf{Q}_{il} - \mathbf{C}_{il})]$
            \Else: for $i\in\tilde{\U}$ and $l\in\D_i$:
                \State $\Gamma_{il} \gets \frac{\displaystyle \tr(\mathbf{S}_{il}) - \sum_{i' \in \tilde{\U} \setminus\{i\}} \frac{\tr(\mathbf{S}_{il} \mathbf{S}_{i'l})}{\tr(\mathbf{S}_{i'l})}}{\displaystyle \sum_{i'\in\tilde{\U}} \frac{p_{i'}}{Mp_i} \tr(\mathbf{C}_{i'l}) + \sigma_\mathrm{n}^2}$
            \EndIf
            %\State Select subarrays for all $i\in\tilde{\U}$:\; $\D_i^{(t)} \gets \mathrm{SubarraySelection}(\{\mathrm{SINR}_{il}\}_{l\in\mathcal{L}},L_i)$
            %\State $\{\mu_{il}\}_{l\in\D_i} \gets \mathbf{OptimalWeightsComputation}(\{\mathrm{SINR}_{il}\}_{l\in\mathcal{L}})$, for $i\in\tilde{\U}$.
            \State $\mu_{il}^{(t)} \gets \displaystyle \Gamma_{il} \Big/ \sum\limits_{l'\in\D_i} \Gamma_{il'}$, $i\in\tilde{\U}$ and $l\in\D_i$
            \State $\eta_{\min}^{(t)} \gets \displaystyle \min_{i\in\tilde{\U}} \left( 1-\frac{\tau}{\tau_\mathrm{c}} \right) \log_2 \left( 1 + \sum_{l\in\D_i} \Gamma_{il} \right)$
        \EndFor
        \State $t_k \gets \displaystyle \operatorname*{arg\,max}_{t\in\{1,\ldots,\tau_\mathrm{p}+1\}} \eta_{\min}^{(t)}$ and $\eta_{\min} \gets \eta_{\min}^{(t_k)}$
        \If{$\eta_{\min} \geq \eta_{\mathrm{th}}$}
            \State $\U \gets \tilde{\U}$ and $\tau_\mathrm{p} \gets \max(\tau_\mathrm{p}, t_k)$
            %\State Update selected subarrays for all $i\in\mathcal{U}$: $\D_i\gets\D_i^{(t_k)}$
            \State $\mu_{il}\gets\mu_{il}^{(t_k)}$,\; $i\in\mathcal{U}$ and $l\in\D_i$
        \EndIf
    \EndWhile
\end{algorithmic}
\end{algorithm}

For each candidate pilot, Algorithm \ref{alg:cent:SINRbased} evaluates the \glspl{SINR} of the \glspl{UE} $\tilde{\U}$ using the ergodic approximation in \eqref{eq:cent:SINRk:erg} when $\lvert\tilde{\U}\rvert > U_\mathrm{switch}$ and the asymptotic approximation \eqref{eq:cent:SINRk:asy} otherwise (lines 6–14), and then it estimates the minimum \gls{SE} based on these \glspl{SINR} (line 15). This conditional switching rule improves the accuracy of the \gls{SINR} estimates by selecting the analytical expression that is known to be more reliable for each loading regime. Its value is chosen based on the numerical analysis in Section \ref{subsec:SINR:accuracy}.

Similarly, Algorithm \ref{alg:dist:SINRbased} evaluates the resulting \glspl{SINR} for each candidate pilot using \eqref{eq:dist:SINRkl:erg} when $\lvert\tilde{\U}\rvert > U_\mathrm{switch}$ and \eqref{eq:dist:SINRkl:asy} otherwise (lines 6–14). Since optimal weighting is employed, line 16 approximates the global \gls{SINR} as the sum of the local \glspl{SINR} at the selected subarrays, following Proposition \ref{prop:dist:globalSINRapprox_and_optimalweighting}.

%---------------------------
\subsection{Performance Evaluation}
\label{subsec:SS_US_PA:results}
%---------------------------

{This subsection assesses} the efficacy of the proposed algorithms in maximizing the minimum per-user \gls{SE}, {demonstrating that strategies} based on deterministic \gls{SINR} expressions {achieve performance comparable to those} based on instantaneous \gls{SINR}, {but with significantly} reduced computational complexity. We assume $L=16$, $L_k=4$, and $M=16$, {with the simulation setup following the parameters established in} Section \ref{subsec:SINR:accuracy}. {Since \gls{LSF}-based subarray selection yields performance comparable to the \gls{SINR}-based strategy (as shown in Section \ref{subsec:SINR:accuracy}), it is adopted hereafter for its lower complexity.}

Figure \ref{fig:alg:minSE} {depicts} the minimum per-user \gls{SE} as a function of the number of scheduled users ($U$). {We compare Algorithms \ref{alg:cent:SINRbased} and \ref{alg:dist:SINRbased} (labeled \textit{Max-Min SE (Ana)}) against a numerical benchmark, \textit{Max-Min SE (Num)}. This benchmark follows the same logic as the proposed algorithms but utilizes the numerical mean \gls{SINR} for pilot assignment instead of the derived deterministic approximations. Two baselines serve as references: Algorithm \ref{alg:NMSEbased}, which is labeled \textit{Orthogonal} since it assigns mutually orthogonal pilots by minimizing the maximum channel estimation \gls{NMSE}, and a \textit{Random} pilot assignment. All schemes employ user scheduling based on descending \gls{LSF} gains and \gls{LSF}-based subarray selection. The scenario features $K=16$ \glspl{UE} and a short coherence block ($\tau_\mathrm{c} = 10$), rendering pilot contamination unavoidable.}\footnote{{Results are plotted up to the full load $U=K$ for illustrative purposes. In practice, the algorithm terminates when the \gls{SE} threshold is violated, potentially leaving some \glspl{UE} unscheduled.}}

% \begin{figure}[!ht]
%     \centering
%     \includegraphics[width=0.75\linewidth]{minSE_vs_U.eps}
%     \caption{Minimum per-user SE versus number of scheduled UEs in centralized operation.}
%     \label{fig:alg:cent:minSE_vs_U}
% \end{figure}

\begin{figure}[!ht]
    \centering
    \subfloat[Centralized Operation.]{
    \includegraphics[width=0.48\linewidth]{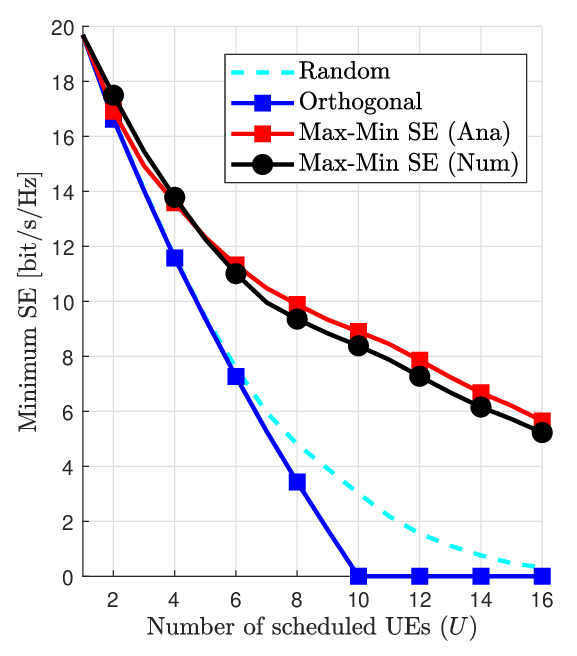}\label{fig:alg:cent:minSE}}
    \hfill
    \subfloat[Distributed Operation.]{
    \includegraphics[width=0.48\linewidth]{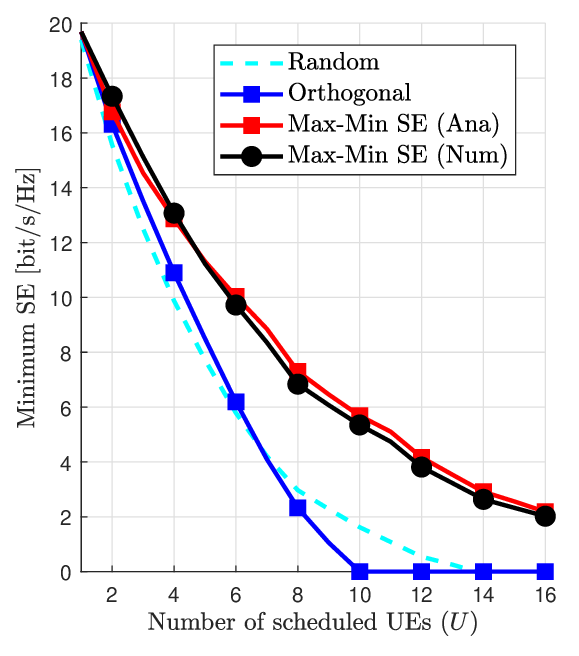}\label{fig:alg:dist:minSE}}
    \caption{Minimum per-user SE versus number of scheduled UEs.}
    \label{fig:alg:minSE}
\end{figure}

{Figures \ref{fig:alg:cent:minSE} and \ref{fig:alg:dist:minSE}, corresponding to centralized and distributed architectures respectively, show that the centralized approach consistently yields a higher minimum \gls{SE}. This is expected, as the \gls{CPU} in a centralized architecture jointly processes signals from all $L_k$ selected subarrays, enabling global interference suppression. Conversely, the distributed architecture relies on local combining, limiting interference mitigation capabilities. Notably, if each \gls{UE} were served by a single subarray ($L_k = 1$), both architectures would converge to identical performance.}

%{Figures \ref{fig:alg:cent:minSE} and \ref{fig:alg:dist:minSE} correspond respectively to the centralized and distributed architectures. Notably, the first one provides higher minimum SE, as expected, since multiple subarrays serve each UE (in this case, $L_k=4$), enabling joint interference supression with the antennas of the selected subarray, which is not possible in the distributed architecture. If $L_k = 1$, the centralized and distributed architectures would achieve the same performance.}

{In both architectures, the \textit{Max-Min SE (Num)} strategy effectively mitigates interference by leveraging the true mean \gls{SINR}; however, it incurs prohibitive computational costs by relying on extensive instantaneous \gls{CSI} processing. In contrast, the \textit{Max-Min SE (Ana)} strategy achieves remarkably similar performance using only deterministic \gls{SINR} approximations, offering a tractable solution. The negligible performance gap suggests that the primary gains in pilot assignment stem from resolving large-scale spatial conflicts, which the derived deterministic expressions capture efficiently without requiring Monte Carlo sampling or instantaneous fading optimization.}

{Furthermore, the proposed strategy rivals the numerical benchmark in maximizing the number of scheduled \glspl{UE} that satisfy a per-user \gls{SE} requirement. For instance, given a $5$ bit/s/Hz threshold in Fig. \ref{fig:alg:cent:minSE}}, the \textit{Random} and \textit{Orthogonal} baselines schedule only 7 \glspl{UE}, whereas both \textit{max-min} algorithms successfully schedule all 16 \glspl{UE}. These results highlight the effectiveness of deterministic \gls{SINR}-based resource allocation in enhancing fairness within crowded \gls{XL-MIMO} networks.}

While {the} intelligent strategies maintain robust performance as $U$ increases, the \textit{Random} and \textit{Orthogonal} assignments suffer severe degradation due to uncoordinated pilot reuse and reduced data transmission intervals ({caused by increasing pilot sequence lengths}), respectively.% The narrow gap between the analytical (black) and numerical (red) strategies indicates that, in practical XL-MIMO scenarios, the primary gain in pilot assignment stems from resolving the large-scale spatial conflicts, which the proposed algorithm accomplishes efficiently. The marginal gain from instantaneous fading optimization in the numerical benchmark yields diminishing returns compared to the high computational cost.

{Finally, Figure \ref{fig:alg:comparing_with_exhaustive_search} illustrates the minimum per-user \gls{SE} as a function of the coherence block length ($\tau_\mathrm{c}$) for $K=6$ \glspl{UE}. A \textit{Max-Min SE (Num - Exhaustive Search)} benchmark is introduced to establish the theoretical performance upper bound. Unlike the sequential approach of the \textit{Max-Min SE (Num)} strategy—where pilots are assigned one-by-one to newly scheduled \glspl{UE} without re-evaluating previous assignments—the exhaustive search evaluates all possible pilot combinations to identify the global optimum.}

{Notably, the sequential mechanism employed by the proposed Algorithms \ref{alg:cent:SINRbased} and \ref{alg:dist:SINRbased} proves highly effective, achieving performance nearly identical to the optimal exhaustive search. Moreover, results suggest that the \textit{Max-Min SE (Ana)} approach may approximate the optimal pilot assignment even more closely than the numerical version. This aligns with observations in Figure \ref{fig:alg:minSE}, where the analytical strategy occasionally outperforms the \textit{Max-Min SE (Num)} benchmark, likely because the deterministic expressions effectively capture long-term spatial conflicts without the variance associated with limited instantaneous \gls{SINR} samples.}

%{Finally, Figure \ref{fig:alg:comparing_with_exhaustive_search} plots the minimum per-user SE as a function of the coherence block length ($\tau_\mathrm{c}$), with $K=6$ \glspl{UE}. A new strategy, denoted as \textit{Max-Min SE (Num - Exhaustive Search)}, is added. This also maximizes the minimum SE, by estimating the SE using the numerical mean of the SINRs. However, instead of following the mechanism of \textit{Max-Min SE (Num)}, where UEs are scheduled sequentially, with a assigned pilot not changing when a new UE is added (the algorithm only chooses the best pilot for the new UE: the pilot that maximizes the minimum SE among this new UE and the already scheduled UEs), the \textit{Max-Min SE (Num - Exhaustive Search)} is the performance upper bound as it performs an exhaustive search, testing all possible combinations of pilots for the $K$ UEs, and then choosing the pilot combination that maximizes the minimum per-user SE among the $K$ UEs. Notably, the mechanism of the proposed Algorithms \ref{alg:cent:SINRbased} and \ref{alg:dist:SINRbased} is very effective since it achieves a similar performance of the optimal solution (exhaustive search). Furthermore, it sinalizes that the proposed \textit{Max-Min SE (Ana)} approach, which corresponds to Algorithms \ref{alg:cent:SINRbased} and \ref{alg:dist:SINRbased} (for centralized and distributed architectures, respectively) approaches even more to the optimal pilot assignment since Figure \ref{fig:alg:minSE} shown that it performed even slightly better than the \textit{Max-Min SE (Num)} strategy.}

\begin{figure}[!ht]
    \centering
    \includegraphics[width=0.75\linewidth]{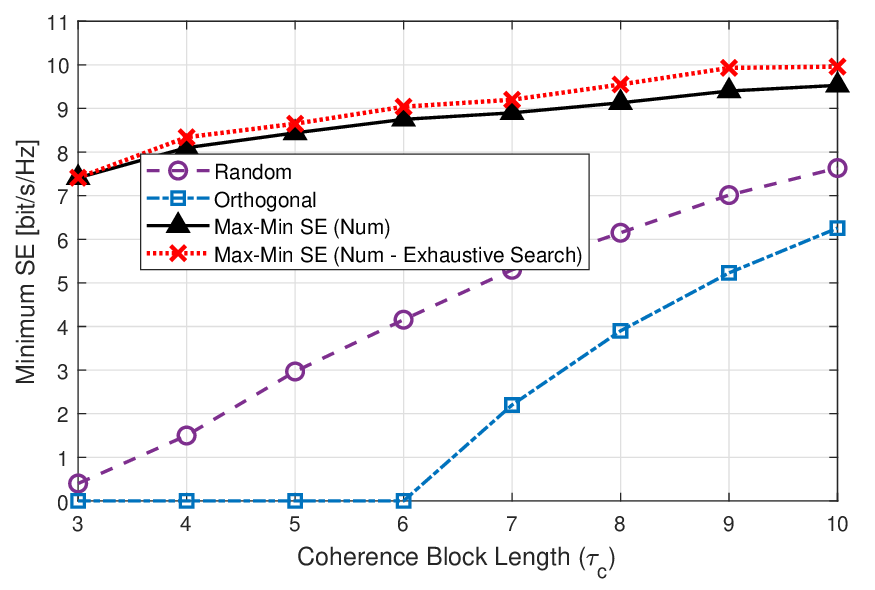}
    \caption{Minimum per-user SE versus coherence block length.}
    \label{fig:alg:comparing_with_exhaustive_search}
\end{figure}

% ===================  \input{conclusion} % =================== 
\section{Conclusion}
\label{sec:conclusion}

%\color{blue}

This paper addressed the critical challenge of resource allocation in XL-MIMO systems, where the reliance on instantaneous CSI for optimization is computationally prohibitive. We introduced novel, closed-form deterministic SINR expressions for both centralized and distributed uplink operations. These expressions, which are valid for Rician fading channels and account for MMSE receive combining and channel estimation, depend exclusively on long-term channel statistics. This key innovation provides a tractable analytical foundation for XL-MIMO system design and optimization, eliminating the need for computationally expensive instantaneous CSI.

Building on these theoretical results, we developed a suite of low-complexity algorithms for joint subarray selection, user scheduling, and pilot assignment. By leveraging the derived statistical SINR approximations, these algorithms are designed to maximize the minimum spectral efficiency among scheduled users, thereby enhancing fairness. Our numerical results validate the high accuracy of the SINR approximations and demonstrate that the proposed statistical-CSI-driven framework achieves performance remarkably close to that of instantaneous-CSI-based benchmarks. The proposed methods effectively exploit the spatial sparsity of user visibility regions to enable more aggressive pilot reuse, significantly improving both fairness and throughput in crowded network scenarios. This work demonstrates that by relying on slowly-varying channel statistics, it is possible to design resource allocation schemes that are both highly efficient and practical for real-world XL-MIMO deployments.

\subsection{Future Research Directions}

The findings of this work open up several promising avenues for future research, including:

\begin{itemize}
\item {\it Hardware Impairments and Channel Aging}: Future studies could extend the deterministic SINR framework to incorporate the effects of practical hardware impairments, such as phase noise and power amplifier nonlinearity, as well as channel aging in high-mobility scenarios. Developing robust resource allocation algorithms would be a significant step toward practical implementation.

\item {\it Downlink Transmission and Precoding Design}: This work focused exclusively on the uplink. A natural extension is to develop a similar statistical-CSI-based framework for the downlink, including the design of low-complexity precoding schemes that can achieve near-optimal performance without requiring instantaneous CSI at the transmitter.

\item {\it Energy Efficiency Optimization}: While the proposed algorithms enhance spectral efficiency and fairness, they do not explicitly consider energy consumption. Future research could focus on developing resource allocation strategies that jointly optimize for spectral and energy efficiency, a critical requirement for sustainable 6G networks.

\item {\it Machine Learning-based Approaches}: The statistical nature of the proposed framework makes it well-suited for machine learning-based enhancements. Deep learning models could be trained to learn the complex relationships between long-term channel statistics and optimal resource allocation decisions, potentially leading to even lower-complexity and more adaptive solutions.

\item {\it Integration with Reconfigurable Intelligent Surfaces} (RIS): The synergy between XL-MIMO and RIS is a promising area of investigation. Future work could explore the joint design of subarray selection in XL-MIMO and passive beamforming at the RIS, leveraging statistical CSI to create a programmable and highly efficient wireless environment.

\item {\it  Near-Field Communications}: As XL-MIMO systems operate in the near-field, future research should more deeply investigate the implications of spherical wave propagation on channel modeling and resource allocation. Extending the deterministic SINR expressions to explicitly account for near-field effects would provide a more accurate analytical foundation for such systems.
\end{itemize}

\color{black}

\appendices
% ===================  \input{proof_globalSINRapprox_and_optimalweighting}
\section{Proof of Proposition \ref{prop:dist:globalSINRapprox_and_optimalweighting}}
\label{app:proof:globalSINRapprox_and_optimalweighting}

When $M \ge U$, the local \gls{ZF} combiner satisfies $\mathbf{v}_{kl}^\mathrm{H} \widehat{\mathbf{h}}_{il} = \delta_{ik}$. Under high \gls{SNR} and small channel estimation errors, the \gls{L-MMSE} combiner converges to the local \gls{ZF}, and thus obeys the approximation $\mathbf{v}_{kl}^\mathrm{H} \widehat{\mathbf{h}}_{il} \approx \delta_{ik}$, yielding the local and global \gls{SINR} approximations given by

\begin{align}
    \Gamma_{kl} &\approx \frac{p_k}{\displaystyle \mathbf{v}_{kl}^\mathrm{H} \left( \sum_{i\in\U} p_i \mathbf{C}_{il} + \sigma_\mathrm{n}^2 \eye_M \right) \mathbf{v}_{kl}},
\label{eq:dist:SINRkl_approx1}
\end{align}

\begin{align}
    \Gamma_k &\approx \frac{p_k}{\displaystyle \sum_{l\in\D_k} \mu_{kl}^2 \mathbf{v}_{kl}^\mathrm{H} \left( \sum_{i\in\U} p_i \mathbf{C}_{il} + \sigma_\mathrm{n}^2 \eye_M \right) \mathbf{v}_{kl}}.
\label{eq:dist:SINRk_approx1}
\end{align}

Corollary \ref{prop:dist:globalSINRapprox_and_optimalweighting} follows directly by rewriting \eqref{eq:dist:SINRk_approx1} and comparing it with the approximation of $\Gamma_{kl}^{-1}$ in \eqref{eq:dist:SINRkl_approx1}:
\begin{equation*}
    \Gamma_k^{-1} \approx \sum_{l\in\D_k} \mu_{kl}^2 \underbrace{ \frac{1}{p_k} \mathbf{v}_{kl}^\mathrm{H} \left( \sum_{i\in\U} p_i \mathbf{C}_{il} + \sigma_\mathrm{n}^2 \eye_M \right) \mathbf{v}_{kl} }_{\displaystyle \Gamma_{kl}^{-1}}.
\end{equation*}

Although \gls{L-MMSE} and local \gls{ZF} combiners perform similarly only under low channel estimation errors, the accuracy of the approximation \eqref{eq:dist:SINRk_from_SINRkl_optweights} may improve as these errors grow. The resulting increase in %the error covariance matrix 
$\mathbf{C}_{il}$ dominates %and stabilizes 
the \gls{SINR} denominator, rendering the neglected interference terms $\mathbf{v}_{kl}^\mathrm{H} \widehat{\mathbf{h}}_{il}$ negligible compared to the $\mathbf{v}_{kl}^\mathrm{H} \mathbf{C}_{il} \mathbf{v}_{kl}$ components.

%The approximation accuracy actually improves under conditions of severe pilot contamination. This behavior arises because the channel estimation error covariance matrix, $\mathbf{C}_{il}$, grows significantly in the presence of pilot reuse. As a result, this deterministic term dominates the denominator of the \gls{SINR} expression, effectively ``stabilizing'' it. Consequently, the terms $\mathbf{v}_{kl}^\mathrm{H} \widehat{\mathbf{h}}_{il}$, which were neglected in the approximations, become negligible compared to the massive estimation error.

The optimal weights for distributed uplink operation are obtained by maximizing the approximate global \gls{SINR} in \eqref{eq:dist:SINRk_from_SINRkl}, given the local \glspl{SINR}. As \(\mu_{kl} \geq 0\) and \(\Gamma_{kl} \geq 0\) for all \(l \in \D_k\), we have $\sum_{l\in\D_k} \mu_{kl}^2 \Gamma_{kl}^{-1} > 0$. Thus, maximizing $\widehat{\Gamma}_k$ is equivalent to solving
\begin{mini!}|s|[2]
{\{\mu_{kl}\}_{l\in\D_k}}
{\sum_{l\in\D_k} \mu_{kl}^2 \Gamma_{kl}^{-1}}
{\label{eq:maxSINRk_problem}}{\label{eq:maxSINRk_objective}}
\addConstraint{\sum_{l\in\D_k} \mu_{kl}}{= 1 \label{eq:maxSINRk_const1}}
\addConstraint{\mu_{kl}}{\geq 0, \quad \forall l \in \D_k \label{eq:maxSINRk_const2},}
\end{mini!}
where constraints \eqref{eq:maxSINRk_const1} and \eqref{eq:maxSINRk_const2} enforce nonnegative, unit-sum weights.
%By ignoring the nonnegativity constraints, believing that they may be inactive, the Lagrangian function is defined as
Neglecting \eqref{eq:maxSINRk_const2}, which will be satisfied a posteriori, the Lagrangian is
\begin{equation}
    \mathcal{L}_k = \sum_{l\in\D_k} \mu_{kl}^2 \Gamma_{kl}^{-1} + \lambda_k \left( \sum_{l\in\D_k} \mu_{kl} - 1 \right),
\label{eq:Lagrangian}
\end{equation}
where \(\lambda_k\) is the lagrangian multiplier. The first-order %\gls{KKT} 
optimality condition gives\cite{boyd}
\begin{equation}
    \frac{\partial \mathcal{L}_k}{\partial \mu_{kl}} = 2 \Gamma_{kl}^{-1} \mu_{kl} + \lambda_k = 0, \quad l \in \D_k,
\label{eq:derivate_Lagrangian_equal0}
\end{equation}
which yields
\begin{equation}
    \mu_{kl} = - \frac{1}{2} \Gamma_{kl} \lambda_k, \quad l \in \D_k.
\label{eq:weights_function_of_multiplier}
\end{equation}

%Since \(\Gamma_{kl} > 0\), feasibility (\(\mu_{kl} \ge 0\)) implies \(\lambda_k < 0\). Enforcing the unit-sum constraint \eqref{eq:maxSINRk_const1} gives
% % 
% \begin{equation}
%     \sum_{l\in\D_k} \mu_{kl}
%     = - \frac{\lambda_k}{2} \sum_{l\in\D_k} \Gamma_{kl}
%     = 1
%     \;\Rightarrow\;
%     \lambda_k = -\frac{2}{\displaystyle \sum_{l\in\D_k} \Gamma_{kl}}.
% \end{equation}

Substituting \eqref{eq:weights_function_of_multiplier} into \eqref{eq:maxSINRk_const1} results in $\lambda_k = -\frac{2}{\sum_{l\in\D_k} \Gamma_{kl}}$. Substituting this into \eqref{eq:weights_function_of_multiplier} completes the proof of \eqref{eq:dist:weights:optimal}, and inserting \eqref{eq:dist:weights:optimal} into \eqref{eq:dist:SINRk_from_SINRkl} yields %the corresponding maximum global \gls{SINR} in 
\eqref{eq:dist:SINRk_from_SINRkl_optweights}.

% ===================  \input{proof_ergoSINRdist}
\section{Proof of Theorem \ref{theo:SINR:ergodic}}
\label{app:proof:SINR_ergo_dist}

%To facilitate the proof, 
We first establish the following lemma regarding the expectation of the inverse of a random matrix.

\begin{lemma}
\label{lemma:inverse_expectation_approximation}
    Let $\mathbf{X}$ be an invertible random matrix with mean $\mathbf{A} = \mathbb{E}\{\mathbf{X}\}$. If the fluctuation $\Delta = \mathbf{X} - \mathbf{A}$ is sufficiently small to satisfy $\lVert \mathbf{A}^{-1}\Delta \rVert < 1$, then the %first-order 
    approximation $\mathbb{E}\{\mathbf{X}^{-1}\} \approx (\mathbb{E}\{\mathbf{X}\})^{-1}$ holds.
\end{lemma}

\begin{proof}
    Using the Neumann series expansion \cite{horn2012matrix}, we have $\mathbf{X}^{-1} = (\mathbf{A} + \Delta)^{-1} = \mathbf{A}^{-1} \sum_{n=0}^{\infty} (-1)^n (\Delta \mathbf{A}^{-1})^n$. For small $\Delta$, truncating the series at $n=0$ and taking the expectation (noting that $\mathbb{E}\{\Delta\} = \mathbf{0}$) yields $\mathbb{E}\{\mathbf{X}^{-1}\} = \mathbf{A}^{-1} + \mathcal{O}(\mathbb{E}\{\|\Delta\|^2\})$, which justifies the approximation.
\end{proof}

The average local \gls{SINR} in \eqref{eq:dist:SINRkl_LMMSE} can be expressed as
\begin{equation}
    \E\{\Gamma_{kl}\} = p_k \tr(\E\{ \mathbf{Z}_{kl}^{-1} \widehat{\mathbf{h}}_{kl} \widehat{\mathbf{h}}_{kl}^\mathrm{H}\}).
\label{eq:dist:E_SINRkl}
\end{equation}

Under typical pilot-reuse patterns (where few or no \glspl{UE} share the pilot of \gls{UE} $k$) and limited \gls{LoS} overlap across subarrays, $\mathbf{Z}_{kl}$ and $\widehat{\mathbf{h}}_{kl}$ are approximately uncorrelated. Thus,% the expectation can be decoupled as
%\footnote{Although $\mathbf{Z}_{kl}$, given by \eqref{eq:dist:Zkl}, only contains the estimated channels of \glspl{UE} other than $k$, it may still be correlated with $\widehat{\mathbf{h}}_{kl}$ especially if \gls{UE} $k$ shares its pilot with other \glspl{UE}. According to \eqref{eq:E_hklest_hilestH}, the estimates $\widehat{\mathbf{h}}_{kl}$ and $\widehat{\mathbf{h}}_{il}$ are correlated if both \glspl{UE} $k$ and $i$ have a \gls{LoS} component to subarray $l$. Furthermore, the correlation between these estimates are more significant if \glspl{UE} $k$ and $i$ share the same pilot sequence, as demonstrated in the first case in \eqref{eq:E_hklest_hilestH}. However, this correlation will be neglected when deriving the \gls{SINR} expression, since the overlap of \gls{LoS} regions of different \glspl{UE} is usually small and user $k$ will normally share its pilot with only a few \glspl{UE}, or even with no other user.}, the random matrix \(\mathbf{Z}_{kl}\) and the channel estimate \(\widehat{\mathbf{h}}_{kl}\) are approximately uncorrelated. Hence,
%
\begin{equation}
    \E\{ \mathbf{Z}_{kl}^{-1} \widehat{\mathbf{h}}_{kl} \widehat{\mathbf{h}}_{kl}^\mathrm{H} \}
    \approx 
    %\E \{ \mathbf{Z}_{kl}^{-1} \} \E\{ \widehat{\mathbf{h}}_{kl} \widehat{\mathbf{h}}_{kl}^\mathrm{H} \}
    %= 
    \E\{ \mathbf{Z}_{kl}^{-1} \} ( \mathbf{Q}_{kl} - \mathbf{C}_{kl} ).
\label{eq:dist:Zkl_hkl_uncorrelated}
\end{equation}

Furthermore, when $U \gg M$ or when subarray selection favors subarrays with a strong \gls{LoS} component to \gls{UE} $k$, the matrix \(\mathbf{Z}_{kl}\) concentrates around its mean \(\overline{\mathbf{Z}}_{kl} = \E\{\mathbf{Z}_{kl}\}\). In this regime, $\mathbf{Z}_{kl}$ exhibits small fluctuations, and Lemma \ref{lemma:inverse_expectation_approximation} implies
\begin{equation}
    \E\{\mathbf{Z}_{kl}^{-1}\} \approx
    \left(\E\{\mathbf{Z}_{kl}\}\right)^{-1} =
    \overline{\mathbf{Z}}_{kl}^{-1}.
\label{eq:dist:inv_exp_approx}
\end{equation}

Substituting 
\eqref{eq:dist:inv_exp_approx} into \eqref{eq:dist:Zkl_hkl_uncorrelated} yields
\begin{equation}
    \E\{\mathbf{Z}_{kl}^{-1} \widehat{\mathbf{h}}_{kl} \widehat{\mathbf{h}}_{kl}^\mathrm{H}\} \approx \overline{\mathbf{Z}}_{kl}^{-1} (\mathbf{Q}_{kl} - \mathbf{C}_{kl}).
\label{eq:dist:Zkl_hkl_uncorrelated_2}
\end{equation}

Inserting \eqref{eq:dist:Zkl_hkl_uncorrelated_2} into \eqref{eq:dist:E_SINRkl} completes the proof of \eqref{eq:dist:SINRkl:erg}. The proof of \eqref{eq:cent:SINRk:erg} follows similar steps.

% ===================  \input{dist_erglocalSINR_comp}
\section{Proof of Corollary \ref{cor:complexity:ergodic}}
\label{app:dist:erglocalSINR:complexity}

Evaluating \eqref{eq:dist:SINRkl:erg} primarily involves computing the trace $\tr(\mathbf{A}^{-1} \mathbf{B})$, where $\mathbf{A} = \overline{\mathbf{Z}}_{kl}$ and $\mathbf{B} = \mathbf{Q}_{kl} - \mathbf{C}_{kl}$. Let $\mathbf{C} = \mathbf{A}^{-1} \mathbf{B}$. We efficiently compute the diagonal elements of $\mathbf{C}$ by solving the linear systems $\mathbf{A} \mathbf{c}_m = \mathbf{b}_m$ for $m = 1, \dots, M$, where $\mathbf{c}_m$ and $\mathbf{b}_m$ are the $m$-th columns of $\mathbf{C}$ and $\mathbf{B}$, respectively.

First, the $\mathbf{LDL}^\mathrm{H}$ decomposition of the Hermitian matrix $\mathbf{A}$ requires $\frac{1}{3}(M^3 - M)$ complex multiplications %, equivalent to $M^3 - M$ real multiplications 
\cite[App. B.1.1]{massivemimobook}. The system $\mathbf{L} \mathbf{D} \mathbf{L}^\mathrm{H} \mathbf{c}_m = \mathbf{b}_m$ is then solved for the $m$-th row of each column $\mathbf{c}_m$ in three steps: (i) %\textbf{Forward substitution:} 
solving $\mathbf{L} \mathbf{z}_m = \mathbf{b}_m$ for each $m$ requires a total of $M \times \frac{1}{2}(M^2 - M)$ complex multiplications; (ii) %\textbf{Diagonal scaling:} 
solving $\mathbf{D} \mathbf{y}_m = \mathbf{z}_m$ %for $\mathbf{y}_m$ 
involves $M$ complex-by-real divisions per column, \textit{i.e.}, $M \times 2M = 2M^2$ real divisions; and (iii) %\textbf{Backward substitution (partial):} 
to obtain the diagonal element $[\mathbf{c}_m]_m$, the system $\mathbf{L}^\mathrm{H} \mathbf{c}_m = \mathbf{y}_m$ is solved from the $M$-th row up to the $m$-th row, which requires $%\sum_{j=m+1}^{M} (M-j+1) 
\sum_{n=0}^{M-m}n 
= \frac{(M-m)(M-m+1)}{2}$ complex multiplications. Hence, the total cost for all $m=1, \dots, M$ is $\sum_{m=1}^M \frac{(M-m)(M-m+1)}{2} = \frac{M^3 - M}{6}$ complex multiplications%, or $\frac{1}{2}(M^3 - M)$ real multiplications
.
%since only the $m$-th element of $\mathbf{c}_m$ is needed for the trace, solving the $m$-th row of $\mathbf{L}^\mathrm{H} \mathbf{c}_m = \mathbf{y}_m$ requires $(M - m)$ complex multiplications. Summing over all $m$ yields a total of $\sum_{m=1}^M (M-m) = \frac{1}{2}(M^2 - M)$ complex multiplications, or $\frac{3}{2}(M^2 - M)$ real multiplications.

Assuming each complex multiplication is implemented via three real multiplications and each real division is computationally equivalent to a real multiplication \cite[App. B.1.1]{massivemimobook}, the total complexity of computing $\Gamma_{kl}^\mathrm{erg}$ across the $L_k$ subarrays that serve \gls{UE} $k$ is equivalent to $\mathcal{C}_k^\mathrm{erg} = L_k [M^3 - M + \frac{3}{2}(M^3 - M^2) + 2M^2 + \frac{1}{2}(M^3 - M)]$
real multiplications, which simplifies to \eqref{eq:dist:SINRkl:erg:complexity}. The proof of \eqref{eq:cent:SINRk:erg:complexity} follows similar steps, by substituting the subarray antenna dimension $M$ with the total number of serving antennas $ML_k$.

% \begin{align*}
%     \mathcal{C}_k^\mathrm{erg} &= L_k \Big( M^3 - M + \frac{3}{2}(M^3 - M^2) \nonumber \\
%     &\quad + 2M^2 + \frac{1}{2}(M^3 - M) \Big)
%     %\\
%     %&= \left(3M^3 + \frac{1}{2}M^2 - \frac{3}{2}M \right) L_k
% \end{align*}

% ===================  \input{proof_asymSINRdist}
\section{Proof of Theorem \ref{theo:SINR:asymptotic}}
\label{app:proof:SINR_asym_dist}

Under the assumption of spatially uncorrelated \gls{NLoS} components, where $\mathbf{C}_{il} = \frac{1}{M} \tr(\mathbf{C}_{il}) \mathbf{I}_M$, the instantaneous \gls{SINR} approximation in \eqref{eq:dist:SINRkl_approx1} reduces to\footnote{This remains accurate for correlated \gls{NLoS} components in the absence of pilot reuse. This is because the eigenvalues of the matrix term within parentheses in \eqref{eq:dist:SINRkl_approx1} are closely clustered without pilot reuse. Thus, the \gls{SINR} denominator presents negligible variations when replacing $\mathbf{C}_{il}$ with $\frac{1}{M} \tr(\mathbf{C}_{il}) \mathbf{I}_M$.
%However, if pilot reuse occurs, some eigenvalues of the matrix within parentheses in \eqref{eq:SINR_kl_approx1} become very close to zero. Consequently, replacing $\mathbf{C}_{il}$ by $\frac{1}{M} \tr(\mathbf{C}_{il}) \eye_M$ introduces only minor variations in the matrix within parentheses. Nevertheless, these small variations can significantly impact the denominator of the local \gls{SINR}.
}
\begin{equation}
    \Gamma_{kl} \approx \frac{p_k}{ \left( \displaystyle\sum_{i\in\U} \frac{p_i}{M} \tr(\mathbf{C}_{il}) + \sigma_\mathrm{n}^2 \right) \left\lVert \mathbf{v}_{kl} \right\rVert^2 }.
\label{eq:dist:SINR_kl_approx1}
\end{equation}

The term $\lVert \mathbf{v}_{kl} \rVert^2$ is approximated using the local-ZF combiner, which behaves similarly to the \gls{L-MMSE} combiner at high \gls{SNR} and sufficiently accurate channel estimates. %For notational convenience, assume for a moment that all \glspl{UE} are scheduled, \textit{i.e.}, let $\U=\K$.
Letting $\U=\K$, the local ZF matrix at subarray $l$ is $\mathbf{V}_l = \widehat{\mathbf{H}}_l (\widehat{\mathbf{H}}_l^\mathrm{H} \widehat{\mathbf{H}}_l)^{-1}$, where $\widehat{\mathbf{H}}_l = [\widehat{\mathbf{h}}_{1l}, \dots, \widehat{\mathbf{h}}_{Kl}]$ and $\mathbf{v}_{kl}$ is the $k$-th column of $\mathbf{V}_l$. The block matrix inversion lemma yields \cite[Appendix A]{XL_MIMO_letter}
\begin{align}
    \lVert \mathbf{v}_{kl} \rVert^2
    &= [\mathbf{V}_l^\mathrm{H} \mathbf{V}_l]_{k,k}
    = [(\widehat{\mathbf{H}}_l^\mathrm{H} \widehat{\mathbf{H}}_l)^{-1}]_{k,k}
    \nonumber\\
    &= [\widehat{\mathbf{h}}_{kl}^\mathrm{H} (\mathbf{I}_M - \mathbf{P}_{kl}) \widehat{\mathbf{h}}_{kl}]^{-1},
\label{eq:dist:inv_HH_H}
\end{align}
where $\mathbf{P}_{kl} = \ddot{\mathbf{H}}_{kl} ( \ddot{\mathbf{H}}_{kl}^\mathrm{H} \ddot{\mathbf{H}}_{kl} )^{-1} \ddot{\mathbf{H}}_{kl}^\mathrm{H}$ and $\ddot{\mathbf{H}}_{kl} \in \mathbb{C}^{M \times (K-1)}$ denotes the estimated channel matrix at subarray $l$ excluding \gls{UE} $k$. Substituting \eqref{eq:dist:inv_HH_H} into \eqref{eq:dist:SINR_kl_approx1} results in
\begin{equation}
    \Gamma_{kl} \approx p_k \frac{\widehat{\mathbf{h}}_{kl}^\mathrm{H} (\mathbf{I}_M - \mathbf{P}_{kl}) \widehat{\mathbf{h}}_{kl}}{\displaystyle\sum_{i\in\U} \frac{p_i}{M} \tr(\mathbf{C}_{il}) + \sigma_\mathrm{n}^2}.
\label{eq:dist:SINRkl_approx2}
\end{equation}

\begin{lemma}
    The cross-correlation between the \gls{MMSE} estimates $\widehat{\mathbf{h}}_{kl}$ and $\widehat{\mathbf{h}}_{il}$ is given by \cite[page 107]{CFbook}:
    \begin{align}
        \E\{ \widehat{\mathbf{h}}_{kl} \widehat{\mathbf{h}}_{il}^\mathrm{H} \} = 
        \begin{cases}
            \mathbf{Q}_{kl} - \mathbf{C}_{kl}, & i = k, \\
            \overline{\mathbf{h}}_{kl} \overline{\mathbf{h}}_{il}^\mathrm{H} + \mathbf{A}_{kil}, & t_i = t_k, \\
            \overline{\mathbf{h}}_{kl} \overline{\mathbf{h}}_{il}^\mathrm{H}, & \text{otherwise,}
        \end{cases}
    \label{eq:E_hklest_hilestH}
    \end{align}
    where $\mathbf{A}_{kil} = \sqrt{p_k p_i} \tau_\mathrm{p} \mathbf{R}_{kl} \boldsymbol{\Psi}_{t_kl}^{-1} \mathbf{R}_{il}$. Furthermore, for any matrices $\mathbf{A}, \mathbf{B} \in \mathbb{C}^{M \times M}$ independent of the channel estimates, %By the law of large numbers 
    the quadratic form $\widehat{\mathbf{h}}_{il}^\mathrm{H} \mathbf{B} \mathbf{A} \widehat{\mathbf{h}}_{kl}$ converges to its expectation as $M \to \infty$ \cite[Lemma~4]{ergodicSINRproof}:
    \begin{equation}
        \widehat{\mathbf{h}}_{il}^\mathrm{H} \mathbf{B} \mathbf{A} \widehat{\mathbf{h}}_{kl}
        \xrightarrow[M \to \infty]{a.s.}
        \E\{\widehat{\mathbf{h}}_{il}^\mathrm{H} \mathbf{B} \mathbf{A} \widehat{\mathbf{h}}_{kl}\}
        = \tr(\mathbf{A} \E\{\widehat{\mathbf{h}}_{kl} \widehat{\mathbf{h}}_{il}^\mathrm{H}\} \mathbf{B}).
    \label{eq:law_of_large_numbers}
    \end{equation}
    Applying \eqref{eq:E_hklest_hilestH}, this yields
    \begin{equation}
        \widehat{\mathbf{h}}_{il}^\mathrm{H} \mathbf{B} \mathbf{A} \widehat{\mathbf{h}}_{kl}
        \xrightarrow[M \to \infty]{a.s.}
        \begin{cases}
            \tr[\mathbf{A} (\mathbf{Q}_{kl} - \mathbf{C}_{kl}) \mathbf{B}], & i = k,\\
            \tr[\mathbf{A} (\overline{\mathbf{h}}_{kl} \overline{\mathbf{h}}_{il}^\mathrm{H} + \mathbf{A}_{kil}) \mathbf{B}],
            & t_i = t_k,\\
            \overline{\mathbf{h}}_{il}^\mathrm{H} \mathbf{B} \mathbf{A} \overline{\mathbf{h}}_{kl}, & t_i \neq t_k.
        \end{cases}
    \label{eq:law_of_large_numbers_2}
    \end{equation}
\end{lemma}

The diagonal elements of $\ddot{\mathbf{H}}_{kl}^\mathrm{H} \ddot{\mathbf{H}}_{kl}$ consist of the squared norms $\lVert \widehat{\mathbf{h}}_{il} \rVert^2$ for $i\in\K\setminus\{k\}$, which converge to $\tr\left( \mathbf{Q}_{il} - \mathbf{C}_{il} \right)$ as $M\to\infty$, according to \eqref{eq:law_of_large_numbers_2}. In \gls{XL-MIMO}, the off-diagonal terms are negligible due to the channel spatial non-stationarity.\footnote{From \eqref{eq:law_of_large_numbers_2}, the off-diagonal elements $\widehat{\mathbf{h}}_{il}^\mathrm{H} \widehat{\mathbf{h}}_{i'l}$ ($i \neq i'$) converge to $\overline{\mathbf{h}}_{il}^\mathrm{H} \overline{\mathbf{h}}_{i'l}$ as $M \to \infty$, since \glspl{UE} are typically assigned distinct pilot sequences. In \gls{XL-MIMO}, however, these off-diagonal terms are negligible compared to the diagonal components, since the large array dimensions ensure that \gls{LoS} footprints of different \glspl{UE} have minimal overlap across subarrays. This contrasts with conventional massive MIMO, where \glspl{VR} typically encompass the entire co-located array.} We can thus approximate \eqref{eq:dist:SINRkl_approx2} by
\begin{equation}
    \Gamma_{kl} \approx p_k \frac{\widehat{\mathbf{h}}_{kl}^\mathrm{H} \widehat{\mathbf{h}}_{kl} - \widehat{\mathbf{h}}_{kl}^\mathrm{H} \ddot{\mathbf{H}}_{kl} \mathbf{T}_{kl}^{-1} \ddot{\mathbf{H}}_{kl}^\mathrm{H} \widehat{\mathbf{h}}_{kl}}{\displaystyle\sum_{i\in\U} \frac{p_i}{M} \tr(\mathbf{C}_{il}) + \sigma_\mathrm{n}^2}.
\label{eq:dist:SINRkl:asy:XLMIMOonly_approx}
\end{equation}
where $\mathbf{T}_{kl} = \diag(\{\lVert \widehat{\mathbf{h}}_{il} \rVert^2\}_{i\in\U\setminus\{k\}})$ contains only the diagonal terms of $\ddot{\mathbf{H}}_{kl}^\mathrm{H} \ddot{\mathbf{H}}_{kl}$.

Applying the asymptotic results from \eqref{eq:law_of_large_numbers_2}, the terms in the numerator converge as follows:
\begin{equation}
    \widehat{\mathbf{h}}_{kl}^\mathrm{H} \widehat{\mathbf{h}}_{kl} \xrightarrow{M\to\infty} \tr(\mathbf{Q}_{kl} - \mathbf{C}_{kl})
\label{eq:dist:SINRkl:largeNapprox_1stterm}
\end{equation}
and
\begin{align}
    \widehat{\mathbf{h}}_{kl}^\mathrm{H} \ddot{\mathbf{H}}_{kl} \mathbf{T}_{kl}^{-1} \ddot{\mathbf{H}}_{kl}^\mathrm{H}
    \widehat{\mathbf{h}}_{kl}
    &= \sum_{i\in\U \setminus \{k\}} \frac{ \widehat{\mathbf{h}}_{kl}^\mathrm{H} \widehat{\mathbf{h}}_{il} \widehat{\mathbf{h}}_{il}^\mathrm{H} \widehat{\mathbf{h}}_{kl} }{ \widehat{\mathbf{h}}_{il}^\mathrm{H} \widehat{\mathbf{h}}_{il} }
    \nonumber\\
    &\xrightarrow{M\to\infty} \frac{\tr(\mathbf{X}_{kl} \mathbf{X}_{il})}{\tr(\mathbf{X}_{il})}.
\label{eq:dist:SINRkl:largeNapprox_2ndterm}
\end{align}

Substituting these into \eqref{eq:dist:SINRkl:asy:XLMIMOonly_approx} concludes the proof of \eqref{eq:dist:SINRkl:asy}. The proof of \eqref{eq:cent:SINRk:asy} folllows similar steps.

\end{document}